\documentclass[12pt]{article}

\usepackage[colorlinks, linkcolor=blue, citecolor=blue]{hyperref}           
\usepackage{color}
\usepackage{graphicx,subfigure,amsmath,amssymb,amsfonts,bm,epsfig,epsf,url,dsfont}
\usepackage{amsthm}

\newtheorem{thm}{Theorem}[section]
\newtheorem{lem}[thm]{Lemma}
\newtheorem{assum}[thm]{Assumption}
\newtheorem{proposition}[thm]{Proposition}
\newtheorem{remark}[thm]{Remark}
\newtheorem{definition}[thm]{Definition}

\newtheorem{corollary}[thm]{Corollary}
\newtheorem{model}[thm]{Model}

\usepackage{tikz}
\usepackage{bbm}
\usepackage[small,bf]{caption}
\usepackage{natbib}
\usepackage[top=1in,bottom=1in,left=1in,right=1in]{geometry}
\usepackage{fancybox}
\usepackage{hyperref}
\usepackage{algorithm}
\usepackage{verbatim}

\usepackage{threeparttable}
\usepackage[titletoc,toc]{appendix}

\usepackage{mathtools}

\usepackage{scalerel,stackengine}
\stackMath
\newcommand\reallywidehat[1]{%
\savestack{\tmpbox}{\stretchto{%
  \scaleto{%
    \scalerel*[\widthof{\ensuremath{#1}}]{\kern-.6pt\bigwedge\kern-.6pt}%
    {\rule[-\textheight/2]{1ex}{\textheight}}%WIDTH-LIMITED BIG WEDGE
  }{\textheight}% 
}{0.5ex}}%
\stackon[1pt]{#1}{\tmpbox}%
}

\newcommand*{\rom}[1]{\expandafter\@slowromancap\romannumeral #1@}
\newcommand{\Cov}{\mathrm{Cov}}

\newcommand{\tr}{\mathrm{tr}}

\newcommand{\cN}{\mathcal{N}}

\DeclareMathOperator*{\argmax}{arg max}
\DeclareMathOperator*{\argmin}{arg min}

\usepackage{xspace}

\numberwithin{equation}{section}

\allowdisplaybreaks
\usepackage[utf8]{inputenc}
\usepackage{authblk}
\usepackage{url}
\usepackage[left=1in,right=1in,top=1in,bottom=1in]{geometry}

\usepackage{titlesec}
\titleformat{\part}[display]
  {\centering\bfseries\LARGE}
  {}{0pt}{}
\titlespacing*{\part}{0pt}{1.0em}{1.0em}

\usepackage{booktabs}
\usepackage{array}
\usepackage{enumitem}
\newcolumntype{L}[1]{>{\raggedright\arraybackslash}p{#1}}
\setlist[itemize]{leftmargin=*,nosep,topsep=0pt,parsep=0pt,partopsep=0pt}

\makeatletter
\renewcommand\paragraph{\@startsection{paragraph}{4}{\z@}%
  {1ex}% space before
  {-1em}% run-in
  {\normalfont\normalsize\bfseries}}
\makeatother

\begin{document}
\title{The Interplay of Signal-to-Noise Ratio and Variance Misspecification in Gaussian Mixtures}

\author[1]{Vladimir Serov\thanks{These authors contributed equally to this work.}}
\author[1]{Amnon Balanov$^{*}$\thanks{Corresponding author: \url{amnonba15@gmail.com}}}
\author[1]{Tamir Bendory}

\affil[1]{\normalsize School of Electrical and Computer Engineering, Tel Aviv University, Tel Aviv 69978, Israel}

\maketitle

\begin{abstract}
We study estimation and clustering in Gaussian mixture models under variance misspecification. Observations are generated with true variance $\sigma^2$, while the component means are estimated using a likelihood with variance $\tau^2$, yielding a family of mismatched likelihood functions parameterized by the ratio $\rho=\tau/\sigma$.
We show that the interplay between $\rho$ and the signal-to-noise ratio (SNR) induces a sharp phase diagram. Under correct specification ($\rho=1$), maximum likelihood recovers the true means, independently of the SNR. However, once the model is misspecified, two different regimes emerge. Under under-smoothing ($\rho<1$), the  estimated Gaussian means are displaced from the truth, and in low SNR this discrepancy grows as the SNR decreases: for every fixed $\rho<1$, the squared error scales as $\mathrm{SNR}^{-1}$. Under over-smoothing ($\rho>1$), the fitted likelihood blurs the cluster separation, causing distinct component means to collapse towards the overall mixture center once $\rho^2$ exceeds a threshold of the form $1 + \lambda\,\mathrm{SNR}$, where $\lambda$ depends on the geometry of the true means. %Thus, in low SNR, even mild variance misspecification can qualitatively alter the target itself.
We further show that the hard assignment objective arises as the limit $\tau\to 0$ of the same mismatched likelihood family, and derive corresponding low- and high-SNR results for hard-assignment mean estimation and latent-label recovery. Furthermore, in low SNR, Bayes-optimal clustering is close to random guessing, and the hard-assignment target remains far from the true means. These results show that in low-SNR applications, even mild variance misspecification or hard-assignment procedures can induce substantial bias, whereas in high SNR these effects are largely absent.

\end{abstract}

\newpage
\tableofcontents

\newpage

\section{Introduction}
\label{sec:intro}

\subsection{Problem description}

Gaussian mixture models (GMMs) are a canonical framework for latent-variable inference~\cite{mclachlan2019finite,bishop2006pattern}. Two fundamental tasks arise in these models: estimating the mixture parameters, such as the component means, and inferring the latent label of each sample, a task often referred to as clustering~\cite{fraley2002model, celeux1995gaussian}. In many applications, however, these tasks are carried out using objectives that are not perfectly matched to the true data-generating model. This issue is typically referred to in the literature as model misspecification~\cite{white1982maximum}. 

A particularly important source of mismatch is the noise variance, especially in high-noise settings. In practice, the variance is often treated as known and fixed, even though it may only be approximately calibrated or may differ from the true variance. Since the  presumed variance directly shapes the likelihood landscape, even mild misspecification can alter the population parameter targeted by the estimator. This raises the following fundamental question:
\begin{quote}
\emph{How does variance misspecification affect mean estimation and clustering in GMMs?}
\end{quote}

To study this question, we consider the isotropic equal-weight GMM. We observe $n$ i.i.d.\ samples $y_1,\dots,y_n\in\mathbb{R}^d$, each generated by first drawing a latent label $L\in\{0,\dots,K-1\}$ uniformly at random and then sampling
\begin{align}
    L \sim \mathrm{Unif}(\{0,1,\ldots,K-1\}),
    \qquad
    Y \mid (L=\ell)\sim \mathcal{N}(\mu_\ell^\star,\sigma^2 I_d),
\label{eq:intro_gmm}
\end{align}
where $\bm{\mu}^\star=(\mu_0^\star,\dots,\mu_{K-1}^\star)$ denotes the unknown component means and $\sigma^2$ is the true noise variance. Under the model in \eqref{eq:intro_gmm}, the marginal distribution of $Y$ is the equal-weight GMM
\begin{align}
    p_{\bm{\mu}^\star,\sigma}(y)
    \triangleq
    \frac{1}{K}\sum_{\ell=0}^{K-1} \varphi_d(y;\mu_\ell^\star,\sigma^2 I_d),    \label{eqn:def_marginal_dist}
\end{align}
where $\varphi_d(y;m,\Sigma)$ denotes the density of the $d$-dimensional Gaussian distribution $\mathcal{N}(m,\Sigma)$ evaluated at $y\in\mathbb{R}^d$. Given the observations $y_1,\dots,y_n$, we distinguish between two inferential tasks that are often intertwined in the GMM literature: (i) \emph{component mean estimation}, namely recovery of the component means $\bm{\mu}^\star$; and (ii) \emph{clustering}, or latent-label recovery, namely, inference of the hidden labels $L_1,\dots,L_n$~\cite{jain2010data}.

\subsection{Estimation under variance misspecification}

In what follows, we focus on the first task and develop a theory of mean estimation in GMMs under variance misspecification; we return to clustering later in the work. Specifically, rather than restricting attention to the correctly specified likelihood, we consider a broader family of variance-mismatched objectives in which estimation is performed with an \emph{algorithmic} variance $\tau^2$ that may differ from the true variance $\sigma^2$. In the terminology of misspecified likelihood theory~\cite{white1982maximum}, this means that the fitted model is not perfectly matched to the true data-generating distribution. At the population level, this leads to the mismatched negative log-likelihood
\begin{align}
    \mathcal{L}_\tau(\bm{\mu}; \bm{\mu^\star}) \triangleq -\mathbb{E}_{Y\sim p_{\bm{\mu}^\star,\sigma}} \big[\log p_{\bm{\mu},\tau}(Y)\big].
    \label{eq:def_L_tau}
\end{align}
Thus, $\mathcal{L}_\tau(\bm{\mu}; \bm{\mu^\star})$ evaluates data generated from the true mixture $p_{\bm{\mu}^\star,\sigma}$~\eqref{eqn:def_marginal_dist} using a fitted GMM $p_{\bm{\mu},\tau}$ with variance $\tau^2$. The corresponding population minimizer is any point
$\bm{\mu}_\tau^\star \in \argmin_{\bm{\mu}} \mathcal{L}_\tau(\bm{\mu}; \bm{\mu^\star})$.
The special case $\tau=\sigma$ recovers the population maximum-likelihood target under the correctly specified model.

A natural scale-free parameter is the mismatch ratio
\begin{align}
    \label{eq:rho-def}
    \rho \triangleq \frac{\tau}{\sigma}.
\end{align}
This ratio organizes the estimation problem into three qualitatively distinct regimes: $\rho=1$ corresponds to the correct specification, $\rho<1$ to an \emph{under-smoothed} regime in which the fitted model is sharper than the truth, and $\rho>1$ to an \emph{over-smoothed} regime in which the fitted model is more diffuse. These distinctions become especially important in low SNR, a regime that arises naturally in computational imaging and structural biology, including single-particle cryogenic electron microscopy (cryo-EM)~\cite{lyumkis2019challenges,bendory2020single,scheres2012relion}. In high SNR, each observation is sufficiently informative that moderate variance mismatch typically has only a limited effect on the population target. In low SNR, by contrast, variance misspecification can qualitatively reshape the population landscape. One of the main findings of this work is that under-smoothing displaces the population minimizer away from the ground-truth means, whereas over-smoothing leads to collapsed configurations in which the estimated means merge toward the overall mixture center.

The following informal theorem summarizes this phase diagram. Formal statements are provided in Section~\ref{sec:main-results}.

\begin{thm}[Informal theorem: variance-mismatch phase diagram]
Assume the data are generated from the isotropic $K$-component GMM~\eqref{eq:intro_gmm}, while estimation is performed using the variance-mismatched population objective $\mathcal{L}_\tau$~\eqref{eq:def_L_tau} with algorithmic variance $\tau^2$. Then, the population landscape exhibits the following qualitatively distinct regimes:
\begin{enumerate}
    \item \emph{Correct specification ($\rho=1$).}
    The population minimizer coincides with the ground-truth means, up to permutation.

    \item \emph{Under-smoothed regime ($\rho<1$).}
    The population minimizer is biased away from the truth. This effect is especially pronounced in low SNR: for every fixed $\rho<1$, the population mismatch remains nonvanishing  and the mean-squared-error (MSE) satisfies
    \begin{align}
        \mathrm{MSE}\left(\bm{\mu}_{\tau}^\star,\bm{\mu}^\star\right) \gtrsim \sigma^2  \asymp \mathrm{SNR}^{-1}.       \label{eqn:MSE_lower_bound_under_smooth}
    \end{align}
    By contrast, in high SNR the effect of under-smoothing becomes much less significant.

    \item \emph{Over-smoothed regime ($\rho>1$).}
    When the fitted variance is larger than the true variance, the model tends to blur the cluster structure. As a result, the population objective favors configurations in which the estimated means, $\bm{\mu}_{\tau}^\star$, merge near the overall mixture center, and such merged solutions can become local minima. In the canonical symmetric two-component case, this occurs through a phase transition at
    \begin{align}
        \rho^2 > 1 + \mathrm{SNR}.        \label{eqn:phase_transition_th}
    \end{align}
\end{enumerate}
\end{thm}

Figure~\ref{fig:3} visualizes this phase diagram in the $(\mathrm{SNR},\rho^2)$ plane. The line $\rho=1$ corresponds to the correctly specified model. Moving downward to $\rho<1$ enters the under-smoothed regime, in which low-SNR population drift occurs; see~\eqref{eqn:MSE_lower_bound_under_smooth}. Moving upward to $\rho>1$ enters the over-smoothed regime, where the fitted likelihood becomes increasingly diffuse and, beyond the phase-transition threshold~\eqref{eqn:phase_transition_th}, favors collapsed configurations in which distinct means merge.
A key practical consequence is that, in low SNR, variance misspecification is not a benign modeling error: \emph{even a mild mismatch in the fitted variance can substantially shift the population target away from the ground-truth means.}
%The line $\rho=1$ corresponds to correct specification. Moving downward corresponds to under-smoothed fitting ($\rho < 1$), whereas moving upward corresponds to over-smoothed fitting ($\rho > 1$). 

\begin{figure*}[t!]
  \centering
  \includegraphics[width=0.6\textwidth]{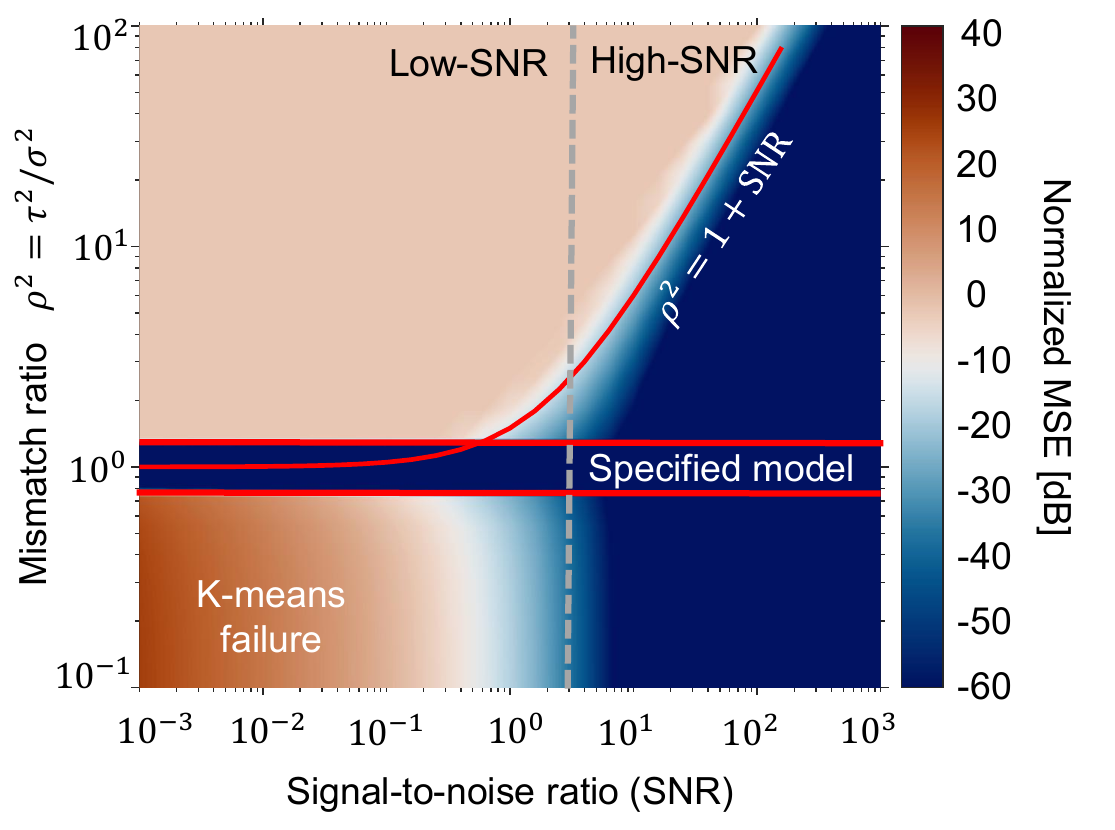}
  \caption{\textbf{Variance-mismatch phase diagram for component mean estimation.}
  Heat map of the normalized mean-squared error (MSE, in dB) of the population means estimator as a function of the signal-to-noise ratio (SNR, horizontal axis) and the variance-mismatch ratio $\rho^2=\tau^2/\sigma^2$ (vertical axis). The true data-generating variance is $\sigma^2$, while $\tau^2$ is the algorithmic variance used by the fitted objective. The line $\rho=1$ corresponds to correct specification, the limit $\tau\to0$ corresponds to hard assignment, and large $\rho$ corresponds to over-smoothed fitting. The diagram highlights two distinct mismatch-induced failure modes: low-SNR population drift in the under-smoothed regime, and collapse of component means to the mixture center in the over-smoothed regime above the phase-transition threshold~\eqref{eqn:phase_transition_th}.}
  \label{fig:3}
\end{figure*}

\subsection{Hard assignment and \texorpdfstring{$k$}{k}-means}
We now turn to a particularly important special case of the variance-mismatch family, namely, the hard-assignment regime. In many applications, component mean estimation is performed by iterative hard-assignment procedures: at each iteration, every observation is assigned to a single component based on the current center, and the means are then updated from the resulting partition. Algorithms of this type are attractive and widely used because of their simplicity and computational efficiency, and are closely related to classical clustering methods such as $k$-means~\cite{mcqueen1967some,lloyd1982least,forgy1965cluster}.

From the perspective developed above, hard assignment is not a separate estimation principle, but rather the zero-temperature limit $\tau \to 0$ of variance-mismatched likelihood fitting. This interpretation will be central in what follows: it places hard-assignment methods within the same continuum as under-smoothed likelihood-based estimation, and allows us to analyze their behavior through the common parameter $\rho=\tau/\sigma \to 0$; see Proposition~\ref{prop:hard_assignment_limit} for a formal statement. In addition, this interpretation leads to a coupled analysis of two problems: latent-label recovery (clustering) and component mean estimation. The key point is that both are governed by the same nearest-center geometry, and their behavior depends strongly on the SNR. In high SNR, hard assignments are typically reliable and the resulting estimated means remain close to the truth; in low SNR, label recovery becomes intrinsically unstable and the hard-assignment population target itself can drift away from the ground truth. Table~\ref{tab:snr_contribution_gmm_compact} summarizes these two regimes and their main consequences.

\begin{table*}[t]
\centering
\begingroup
\setlength{\tabcolsep}{3.5pt}
\renewcommand{\arraystretch}{1.12}
\footnotesize

\caption{\textbf{Low and high-SNR summary for clustering and hard-assignment mean estimation in GMMs}}
\label{tab:snr_contribution_gmm_compact}

\begin{tabular}{|L{2.0cm}||L{6.6cm}|L{6.6cm}|}
\hline
\textbf{Regime} &
\textbf{Latent-label clustering} &
\textbf{Mean estimation} \\
\hline\hline

\textbf{Low SNR } ($\mathrm{SNR} \ll 1$) &
Bayes misclassification error is near random guessing (Proposition~\ref{prop:low_snr_lb_classification}; Corollary~\ref{cor:upper-bound-mutual-information}).
\[ 
    P_{\mathrm{err}}^\star \ge 1-\frac{1}{K} - \frac{1}{2}\sqrt{\mathrm{SNR}}
\]

&
Hard-assignment population optimizer deviates from ground-truth by at least order $\sigma^2$ (Theorem~\ref{thm:lowSNR_under_smoothed_fixed_rho}). 
\[
    \mathrm{MSE}\!\left({\bm{\mu}}_{\mathrm{HA}}^\star(\sigma),\bm{\mu}^\star\right) \gtrsim \sigma^2 \asymp \mathrm{SNR}^{-1}
\]

\\
\hline

\textbf{High SNR} ($\sigma \ll \Delta_{\min}$) &
Bayes misclassification error decays exponentially in $\Delta_{\min}^2/\sigma^2$ (Proposition~\ref{prop:highSNR_exp_bound}). 
\[
    P_{\mathrm{err}}^\star \;\le\; \frac{K-1}{2}\exp\!\left(-\frac{\Delta_{\min}^2}{8\sigma^2}\right)
\]
&
Mean estimation becomes consistent and converges exponentially fast to ground-truth means (Corollary~\ref{prop:highSNR_mean_upper}).
\[
     \mathrm{MSE} \lesssim K^3\big(\Delta_{\max}^2+\sigma^2 d\big)\, \exp\!\left(-\frac{\Delta_{\min}^2}{8\sigma^2}\right)
\]
 
\\
\hline

\end{tabular}
\endgroup
\end{table*}

\begin{figure}
  \centering  \includegraphics[width=0.9\textwidth]{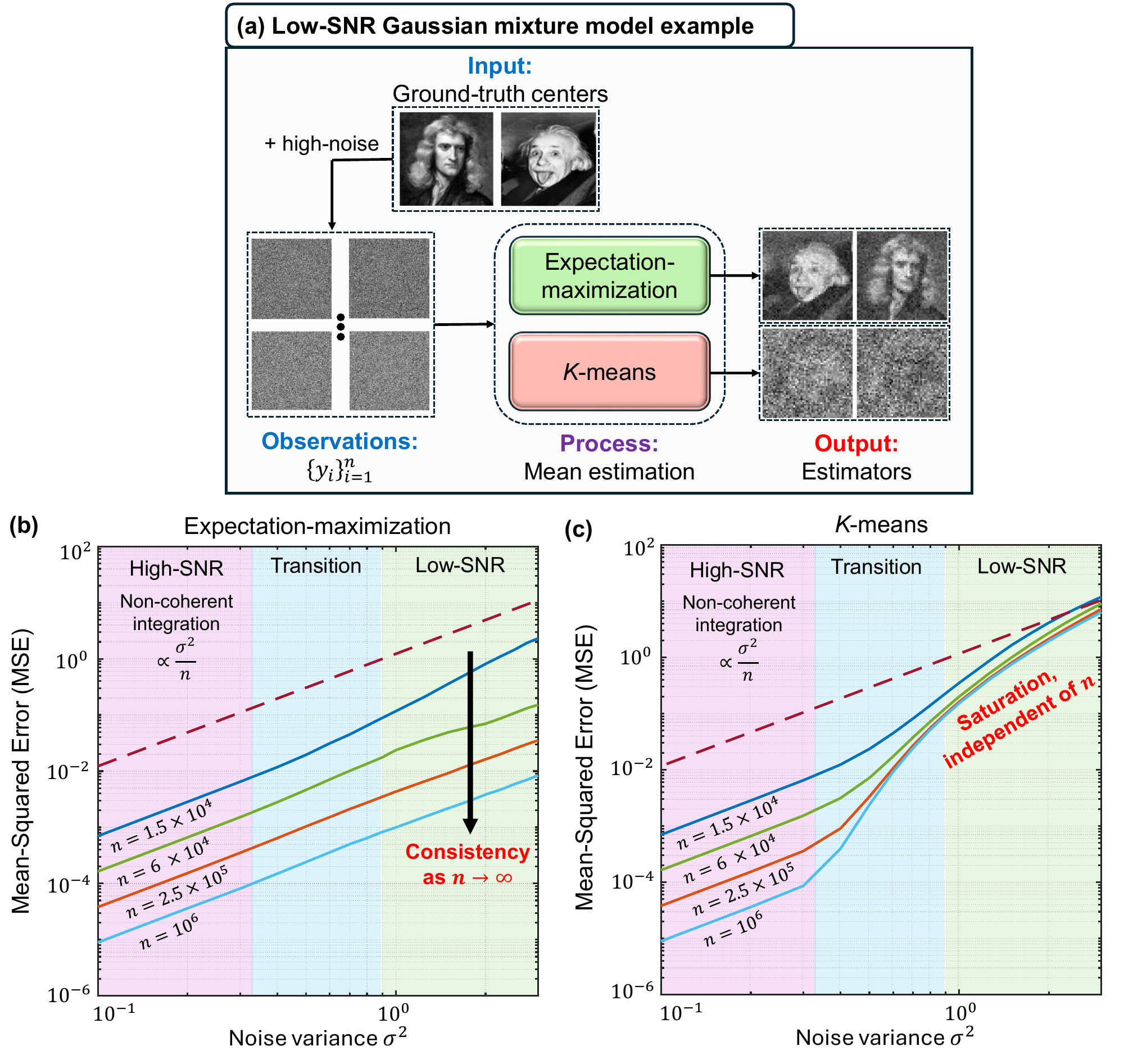}
    \caption{\textbf{Low-SNR Gaussian mixture mean estimation: maximum-likelihood versus hard assignment.}
    \textbf{(a)} Illustration of the low-SNR setting. Observations are generated by adding strong Gaussian noise to samples from a two-component GMM with ground-truth means illustrated by portrait images. The goal is to estimate the component means. In this example, maximum-likelihood estimation, implemented via the expectation-maximization (EM) algorithm, recovers accurate estimates of the means, whereas a hard-assignment-based procedure implemented via $k$-means fails and returns noise-like estimates. Simulation parameters: $\mathrm{SNR}=10^{-6}$, image size $50\times 50$, sample size $n=10^7$, and 300 iterations.
    \textbf{(b)} Maximum-likelihood estimation. MSE of the estimated means versus the noise variance $\sigma^2$ for several sample sizes $n$ (log-log scale). The shaded regions indicate high-SNR, transition, and low-SNR regimes. In the high-SNR regime, the error follows the non-coherent integration scaling $\propto \sigma^2/n$ (dashed guideline). In the low-SNR regime, increasing $n$ continues to reduce the error, consistent with the fact that the likelihood-based population target remains the true mean configuration under correct specification.
    \textbf{(c)} Hard assignment. The same metric and regimes for a hard-assignment-based estimator implemented via $k$-means. In contrast to maximum-likelihood estimation, the error exhibits \emph{low-SNR saturation}: for sufficiently large $\sigma^2$, increasing $n$ yields no further improvement. As explained in the paper, this behavior reflects the interaction between unreliable per-sample label recovery and the bias induced by variance misspecification in the hard-assignment limit. Simulation parameters: image size $15\times 15$, 100 iterations, ground-truth initialization, and 50 Monte Carlo trials.}    \label{fig:1}
\end{figure}

Figures~\ref{fig:1} and~\ref{fig:2} illustrate these phenomena in the two-component case ($K=2$). Figure~\ref{fig:1}(a) shows a low-SNR mixture in which individual observations are visually uninformative, even though the underlying estimation problem remains well posed. Figure~\ref{fig:1}(b) shows that correctly specified likelihood-based estimation, implemented here via expectation-maximization (EM), continues to improve as the sample size increases. Figure~\ref{fig:1}(c), by contrast, shows that a hard-assignment-based procedure exhibits a low-SNR saturation effect: beyond a certain point, increasing $n$ leads to little further improvement.

Figure~\ref{fig:2} explains this discrepancy from two complementary perspectives. Panel~(a) shows that, in low SNR, the Bayes error for latent-label recovery approaches random guessing. Panel~(b) shows that the population target associated with hard assignment develops a nonvanishing mismatch from the true means. Taken together, these figures suggest that the low-SNR failure of hard-assignment methods is not merely an optimization artifact, but rather a consequence of the interaction between label unrecoverability and the population bias induced by variance misspecification. 

\begin{figure*}[t!]
  \centering  \includegraphics[width=0.9\textwidth]{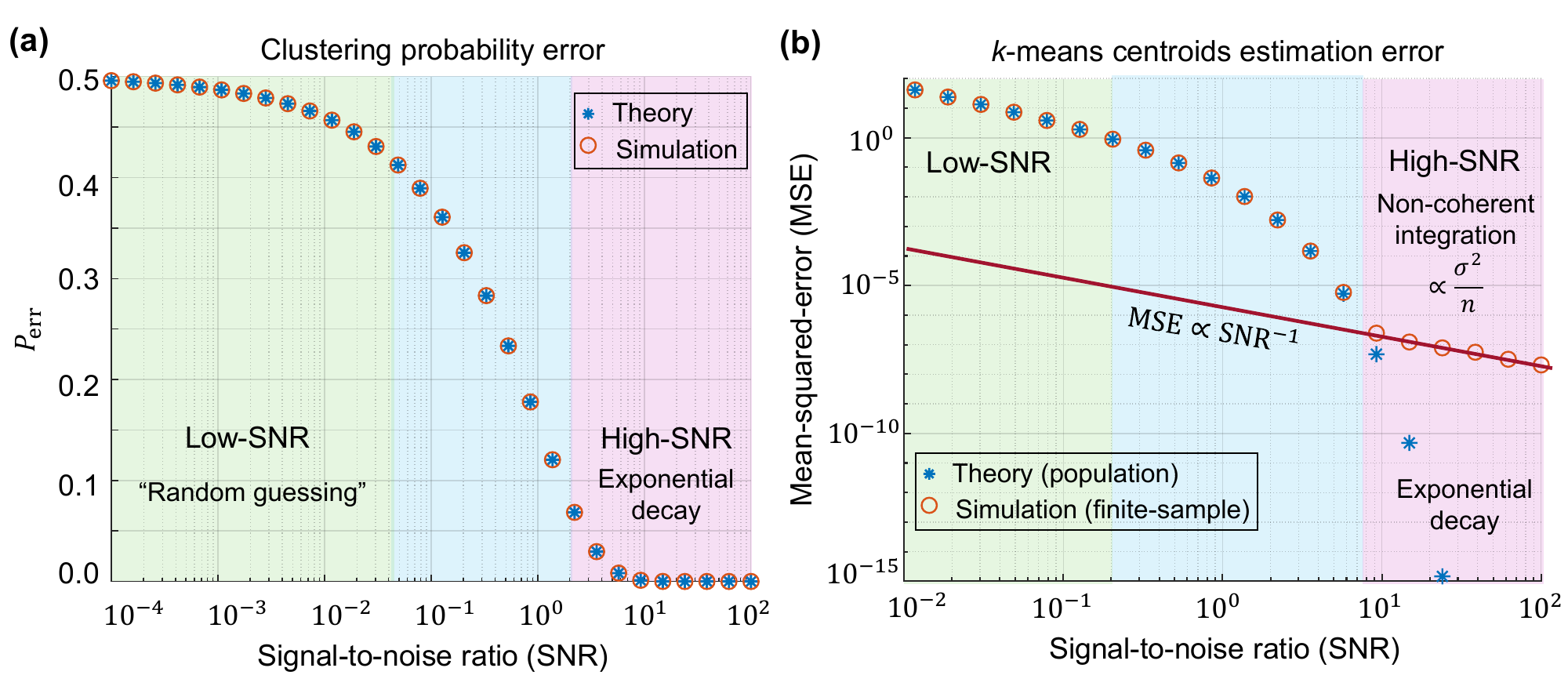}   
  \caption{\textbf{Clustering vs. mean estimation across SNR for a two-component GMM ($K = 2)$.}
  (a) \textbf{Clustering probability of error.} The misclassification probability $P_{\mathrm{err}}$ (Bayes-optimal / theoretical prediction in blue, simulation in orange) as a function of the SNR. In the low-SNR regime, $P_{\mathrm{err}}$ approaches $1/2$, i.e., essentially random guessing, indicating that reliable per-sample label recovery is information-theoretically impossible. In the high-SNR regime, $P_{\mathrm{err}}$ decays rapidly (exponentially in SNR), yielding accurate clustering.
  (b) \textbf{Hard-assignment mean estimation error.} MSE of the $k$-means centers versus SNR, comparing the population-level prediction (blue markers) to finite-sample simulations (orange markers). The shaded regions indicate low-SNR, transition, and high-SNR regimes. In the low-SNR regime, MSE does \emph{not} follow the  $1/n$ law, but rather saturates for every $n$ as seen in Figure~\ref{fig:1}(c). This emphasizes that hard-assignment mean estimation can remain fundamentally limited, and qualitatively different from maximum-likelihood, precisely in the regime where per-sample label recovery is information-theoretically impossible. Simulation parameters: $n=10^6$ observations, 100 Monte Carlo trials, and signal dimension $d=1$.}    \label{fig:2}
\end{figure*}

\subsection{Related work}
A broader statistical literature studies maximum likelihood estimation under model misspecification. In this setting, the estimator generally converges not to the data-generating parameter itself, but to the Kullback--Leibler projection of the true distribution onto the fitted model class~\cite{white1982maximum}. For GMMs, misspecification may arise from the number of components, mixing weights, covariance structure, or component variances. Prior work has studied the impact of covariance and variance misspecification on estimation and inference in normal and multivariate GMMs~\cite{lo2011bias,heggeseth2013impact}, as well as the behavior of EM under misspecified GMMs~\cite{dwivedi2018theoretical}. Our work focuses on a specific misspecification axis: the fitted variance $\tau^2$ may differ from the true variance $\sigma^2$. We show that the resulting population target depends sharply on both the variance ratio $\rho=\tau/\sigma$ and the SNR, leading to under-smoothed drift, over-smoothed collapse, and the hard-assignment limit as $\tau\to0$.

On the algorithmic side, a large body of work studies $k$-means and Lloyd-type procedures~\cite{lloyd1982least}. The $k$-means objective minimizes the empirical quantization loss, and Lloyd's algorithm is the classical alternating method for this problem: given current centers, it assigns each observation to its nearest center and then updates each center by averaging the assigned points.  Many works analyze Lloyd-type methods and related hard-assignment procedures under strong separation or high-SNR assumptions, where misclassification is rare and hard decisions are reliable~\cite{lu2016statistical,loffler2021optimality,ndaoud2022sharp}.  Complementary recent work has emphasized the opposite regime, showing that assignment-based algorithms can exhibit severe pathologies in pure-noise or high-dimensional high-noise settings. In high-dimensional high-noise regimes, Lloyd's algorithm may suffer from fixed-point failures~\cite{lederman2026catastrophic}. In pure-noise GMM settings, hard-assignment updates can produce estimates strongly correlated with the algorithm's initialization even when the data contain no signal~\cite{balanov2025confirmation}.

On the statistical side, classical results view $k$-means as empirical risk minimization for the quantization loss: empirical minimizers converge, under suitable conditions, to minimizers of the population quantization risk, rather than necessarily to the true component means of an underlying mixture model~\cite{pollard1981strong}. More recent work develops nonasymptotic excess-risk bounds for the same objective~\cite{levrard2015nonasymptotic,bachem2017uniform}. A complementary line of work connects hard-assignment procedures to mixture-model objectives through classification EM, variational formulations, and small-variance asymptotics, where probabilistic mixture objectives reduce to $k$-means-like objectives as the fitted covariance or temperature parameter vanishes~\cite{celeux1992cem,lucke2019k,kulis2012revisiting,jiang2012small,broderick2013mad,barber2012bayesian,mackay2003information}. Finally, the learnability literature emphasizes that mixture parameters can remain estimable even when latent-label recovery is difficult~\cite{kalai2010efficiently,moitra2010settling,vempala2004spectral,azizyan2013minimax}.

Our contribution differs from these works in two main respects. First, we study the full variance-mismatch path indexed by $\rho=\tau/\sigma$ and characterize how it changes the population geometry of component mean estimation. Second, we derive SNR-explicit consequences under the true noisy GMM, thereby connecting variance misspecification, hard assignment, and latent-label recovery within a single framework. Thus, while several existing literatures touch on parts of this picture, none addresses the present question in this unified form.

\section{Preliminaries and theoretical framework}
\label{sec:prelim}

In this section, we introduce the statistical model, notation, and objective functions used throughout the paper. We begin with the isotropic equal-weight GMM and the geometric quantities that govern the problem, and then define the metrics used to compare configurations of component means. Finally, we present the family of variance-mismatched likelihood objectives, and explain how hard assignment arises as the zero-temperature limit of this family.

\paragraph{Notation.}
Throughout, $\xrightarrow[]{\mathcal{D}}$, $\xrightarrow[]{\mathbb{P}}$, $\xrightarrow[]{\mathrm{a.s.}}$, and $\xrightarrow[]{\mathcal{L}^p}$ denote convergence in distribution, in probability, almost surely, and in $\mathcal{L}^p$, respectively.
We write $\|\cdot\|$ for the Euclidean norm on vectors and $\|\cdot\|_F$ for the Frobenius norm on matrices. For an integer $K\ge 1$, we write $[K]\triangleq\{0,1,\dots,K-1\}$. 
We employ the following standard asymptotic notation. For nonnegative sequences ${a_n}$ and ${b_n}$, we write $a_n=O(b_n)$ if there exists a constant $C<\infty$ and $n_0$ such that $a_n\le C b_n$ for all $n\ge n_0$, and $a_n=o(b_n)$ if $a_n/b_n\to 0$. We write $a_n=\omega(b_n)$ if $a_n/b_n\to\infty$, and $a_n=\Theta(b_n)$ if both $a_n=O(b_n)$ and $b_n=O(a_n)$ hold.

\subsection{Gaussian mixture model (GMM)}
\label{subsec:gmm}

We study the problem of component mean estimation under variance misspecification in the isotropic equal-weight GMM. In this model, each observation is generated by first drawing a latent component uniformly at random and then adding isotropic Gaussian noise to the corresponding component mean, as described in the next model. %Throughout the paper, this is the only mixture model under consideration.

\begin{model}[Isotropic equal-weight Gaussian mixture model]
\label{model:gmm_isotropic_equal}
Fix integers $K\ge 2$ and $d\ge 1$. For any collection of means $\bm{\mu}=(\mu_0,\dots,\mu_{K-1})\in(\mathbb{R}^d)^K$ and any variance level $s>0$, define the corresponding equal-weight isotropic GMM density by
\begin{align}
    p_{\bm{\mu},s}(y) \triangleq \frac{1}{K}\sum_{\ell=0}^{K-1}\varphi_d\!\left(y;\mu_\ell,s^2 I_d\right),
    \qquad y\in\mathbb{R}^d.    \label{eq:gmm_isotropic_density}
\end{align}
where $\varphi_d(y;\mu,\Sigma)$ denotes the density of the Gaussian distribution $\mathcal{N}(\mu,\Sigma)$ evaluated at $y\in\mathbb{R}^d$. The observations are generated from ground-truth means $\bm{\mu}^\star=(\mu_0^\star,\dots,\mu_{K-1}^\star)\in(\mathbb{R}^d)^K$ with pairwise distinct entries and true noise level $\sigma>0$, namely, 
\begin{align}
    y_1,\dots,y_n \stackrel{\mathrm{i.i.d.}}{\sim} p_{\bm{\mu}^\star,\sigma}.   \label{eq:gmm_isotropic_model}
\end{align}
Equivalently, conditioned on $L = \ell$, 
\begin{align}
    Y \mid (L=\ell) \sim \mathcal{N}(\mu_\ell^\star,\sigma^2 I_d),
    \qquad \ell\in[K].    \label{eq:gmm_latent_model}
\end{align}
\end{model}

\paragraph{Signal-to-noise ratio and minimum pairwise distance.}
Let $\bar{\mu}^\star \triangleq \frac{1}{K}\sum_{\ell=0}^{K-1}\mu_\ell^\star$ denote the global mixture mean. We measure the SNR by
\begin{align}
    \mathrm{SNR} \triangleq \frac{1}{\sigma^2} \frac{1}{K}\sum_{\ell=0}^{K-1}\|\mu_\ell^\star-\bar{\mu}^\star\|^2,
    \label{eq:snr-def}
\end{align}
and we quantify class separation by the minimum pairwise distance
\begin{align}
    \Delta_{\min} \triangleq \min_{\ell\neq j}\|\mu_\ell^\star-\mu_j^\star\|.
    \label{eq:delta_min_def}
\end{align}
Throughout, we assume that the true means are pairwise distinct, that is, $\Delta_{\min}>0$.

\paragraph{Invariant distance and mean-squared error.}

Since mixture components are identifiable only up to permutation of their labels, all metrics are taken modulo label permutations.

\begin{definition}[Permutation-invariant distance]
\label{def:perm_distance}
Let $\bm{\mu}=(\mu_\ell)_{\ell\in[K]}$ and $\bm{\nu}=(\nu_\ell)_{\ell\in[K]}$ be two configurations in $(\mathbb{R}^d)^K$. Define the permutation-invariant distance by
\begin{align}
    d_{\mathrm{perm}}(\bm{\mu},\bm{\nu}) \triangleq \min_{\pi\in S_K}\big\|\bm{\mu}-\pi\cdot\bm{\nu}\big\|_F,
    \label{eq:perm-distance}
\end{align}
where $S_K$ is the symmetric group on $K$ elements, $\pi\cdot\bm{\nu}$ denotes the permuted tuple $(\nu_{\pi(\ell)})_{\ell\in[K]}$, and $\|\cdot\|_F$ is the Frobenius norm on $(\mathbb{R}^d)^K$.
\end{definition}

Throughout, we denote by $\mathcal{Y}=\{y_i\}_{i=1}^n$ the observations generated from the isotropic equal-weight model~\eqref{eq:gmm_isotropic_model}.

\begin{definition}[Normalized mean-squared error]
\label{def:normalized_mse}
Let $\widehat{\bm{\mu}}(\mathcal{Y})$ be an estimator of the ground-truth configuration $\bm{\mu}^\star$ constructed from the observations $\mathcal{Y}$. We define the normalized MSE by
\begin{align}
    \mathrm{MSE}\bigl(\widehat{\bm{\mu}}(\mathcal{Y}),\bm{\mu}^\star\bigr) \triangleq \frac{1}{\|\bm{\mu}^\star\|_F^2}\, \mathbb{E}\!\left[ d_{\mathrm{perm}}^2\bigl(\widehat{\bm{\mu}}(\mathcal{Y}),\bm{\mu}^\star\bigr) \right],
    \label{eq:mseDef_gmm}
\end{align}
where the expectation is taken with respect to the joint distribution of the observed sample $\mathcal{Y}=\{y_i\}_{i=1}^n$.
\end{definition}

\subsection{Maximum-likelihood estimation objectives}
\label{subsec:likelihood_and_kmeans}

 For a candidate mean configuration $\bm{\mu}=(\mu_0,\dots,\mu_{K-1})\in(\mathbb{R}^d)^K$ and an algorithmic variance level $\tau>0$, define the empirical variance-mismatched negative log-likelihood by
\begin{align}
    \mathcal{L}_{\tau}^{(n)}(\bm{\mu};\mathcal{Y}) \triangleq -\frac{1}{n}\sum_{i=1}^n \log p_{\bm{\mu},\tau}(y_i).
    \label{eq:empirical_loglik}
\end{align}
The associated empirical estimator is then defined as a minimizer
\begin{align}
    \widehat{\bm{\mu}}_{\tau}^{(n)} \in \argmin_{\bm{\mu}\in(\mathbb{R}^d)^K} \mathcal{L}_{\tau}^{(n)}(\bm{\mu};\mathcal{Y}).
    \label{eq:mle_general}
\end{align}

Passing to the population level, the strong law of large numbers implies that, for each fixed $\bm{\mu}$ and $\tau>0$,
\begin{align}
    \mathcal{L}_{\tau}^{(n)}(\bm{\mu};\mathcal{Y})
    \xrightarrow[n\to\infty]{\mathrm{a.s.}}
    \mathcal{L}_{\tau}(\bm{\mu}),
    \label{eq:lln_pointwise}
\end{align}
where $\mathcal{L}_{\tau}$ is the population objective defined in~\eqref{eq:def_L_tau}. Accordingly, the infinite-data target of the empirical estimator is any minimizer of $\mathcal{L}_{\tau}$, which we denote by
\begin{align}
    \bm{\mu}_\tau^\star \in \argmin_{\bm{\mu}\in(\mathbb{R}^d)^K} \mathcal{L}_{\tau}(\bm{\mu}).
    \label{eqn:quasi-MLE}
\end{align}
This is the population quasi-maximum-likelihood estimator associated with the fitted variance $\tau^2$. Equivalently, and as is standard in misspecified likelihood theory~\cite{white1982maximum}, $\bm{\mu}_\tau^\star$ is the Kullback-Leibler projection of the true distribution onto the model class with fitted variance $\tau^2$:
\begin{align}
    \argmin_{\bm{\mu}\in(\mathbb{R}^d)^K}\mathcal{L}_{\tau}(\bm{\mu}) = \argmin_{\bm{\mu}\in(\mathbb{R}^d)^K} \; D_\mathrm{KL}\bigl(p_{\bm{\mu}^\star,\sigma}\,\|\,p_{\bm{\mu},\tau}\bigr).    \label{eq:KL_projection}
\end{align}

\paragraph{Correctly specified likelihood model.}

The special case $\tau=\sigma$ corresponds to the correctly specified model. Then,  $\mathcal{L}_{\sigma}^{(n)}$ is the usual negative log-likelihood, and any minimizer $\widehat{\bm{\mu}}_{\sigma}^{(n)} \in \argmin_{\bm{\mu}\in(\mathbb{R}^d)^K} \mathcal{L}_{\sigma}^{(n)}(\bm{\mu};\mathcal{Y})$ is the ordinary maximum-likelihood estimator. Under correct specification, maximum-likelihood estimation serves as the classical benchmark in parametric inference: under standard regularity conditions, it is consistent, asymptotically normal, and asymptotically efficient. Thus, the correctly specified model serves as the natural reference point against which the variance-mismatched regimes are compared.
In practice, a standard method for minimizing $\mathcal{L}_{\sigma}^{(n)}$ is the expectation-maximization (EM) algorithm~\cite{dempster1977maximum}. %In particular, every EM fixed point is a stationary point of the empirical correctly specified likelihood.

\subsection{Hard assignment as a variance-misspecification limit}
\label{subsec:hard_assignment_limit}

We now make precise the connection between hard assignment and variance misspecification. The fitted variance $\tau^2$ defines a family of likelihood-based objectives $\{\mathcal{L}_\tau\}_{\tau>0}$, with the correctly specified model corresponding to $\tau=\sigma$. The limit $\tau\to 0$ yields the hard-assignment objective, showing that hard assignment is not a separate estimation principle, but the zero-temperature limit of the variance-mismatched likelihood family.

\paragraph{Hard-assignment objective.}
At the empirical level, hard assignment is described by the objective
\begin{align}
    \Phi_n(\bm{\mu};\mathcal{Y}) \triangleq \frac{1}{n}\sum_{i=1}^n \min_{\ell\in[K]}\|y_i-\mu_\ell\|^2,   \label{eq:kmeans_empirical}
\end{align}
with corresponding estimator
\begin{align}
    \widehat{\bm{\mu}}_{\mathrm{HA}}^{(n)} \in \argmin_{\bm{\mu}\in(\mathbb{R}^d)^K} \Phi_n(\bm{\mu};\mathcal{Y}). \label{eq:kmeans_empirical_minimizer}
\end{align}
Its population counterpart is
\begin{align}
    \Phi(\bm{\mu}) \triangleq \mathbb{E}_{Y\sim p_{\bm{\mu}^\star,\sigma}} \Big[\min_{\ell\in[K]}\|Y-\mu_\ell\|^2\Big], \label{eq:kmeans_population}
\end{align}
and the corresponding population hard-assignment target is any minimizer
\begin{align}
    \bm{\mu}_{\mathrm{HA}}^\star \in \argmin_{\bm{\mu}\in(\mathbb{R}^d)^K} \Phi(\bm{\mu}).    \label{eq:kmeans_population_minimizer}
\end{align}
For each fixed $\bm{\mu}$, the strong law of large numbers yields $\Phi_n(\bm{\mu};\mathcal{Y})\xrightarrow{\mathrm{a.s.}}\Phi(\bm{\mu})$, as $n \to \infty$, so $\bm{\mu}_{\mathrm{HA}}^\star$ is the infinite-data target of empirical hard assignment.

To state the limit result cleanly, and to avoid technical issues related to existence of minimizers, we restrict attention to a compact parameter set $\mathcal{U}\subset(\mathbb{R}^d)^K$.

\begin{assum}[Well-posedness and uniqueness up to permutation]
\label{ass:unique_minimizers_up_to_perm}
Assume that $\mathcal{U}\subset(\mathbb{R}^d)^K$ is compact. For every $\tau>0$, suppose that the population variance-mismatched objective $\mathcal{L}_\tau$, defined in~\eqref{eq:def_L_tau}, admits a unique minimizer over $\mathcal{U}$ up to permutation: namely, there exists $\bm{\mu}_\tau^\star\in\mathcal{U}$ such that
\begin{align}
    \argmin_{\bm{\mu}\in\mathcal{U}} \mathcal{L}_\tau(\bm{\mu}; \bm{\mu^\star}) = \{\pi\cdot\bm{\mu}_\tau^\star:\pi\in S_K\}.    \label{eq:unique_Ltau_up_to_perm_assump}
\end{align}
Likewise, suppose that the population hard-assignment risk $\Phi$, defined in~\eqref{eq:kmeans_population}, admits a unique minimizer over $\mathcal{U}$ up to permutation: namely, there exists $\bm{\mu}_{\mathrm{HA}}^\star\in\mathcal{U}$ such that
\begin{align}
    \argmin_{\bm{\mu}\in\mathcal{U}} \Phi(\bm{\mu}) = \{\pi\cdot\bm{\mu}_{\mathrm{HA}}^\star:\pi\in S_K\}.    \label{eq:unique_HA_up_to_perm_assump}
\end{align}
\end{assum}

The next proposition, which is proved in Appendix~\ref{sec:proof_hard_assignment_limit}, shows that hard assignment is the zero-temperature limit of the variance-mismatched likelihood family. In this sense, hard assignment corresponds to fitting an increasingly sharp GMM to data generated at a fixed nonzero noise level $\sigma^2$. This interpretation will be important in the sequel: the behavior of hard assignment in the low- and high-SNR regimes should be understood as an extreme manifestation of the broader variance-mismatch phenomenon.

\begin{proposition}[Hard assignment is the $\tau\to0$ limit]
\label{prop:hard_assignment_limit}
Under Assumption~\ref{ass:unique_minimizers_up_to_perm}, the following holds:
\begin{enumerate}
    \item For every fixed $\bm{\mu}\in\mathcal{U}$ and every $\tau>0$,
    \begin{align}
        \mathcal{L}_\tau(\bm{\mu}; \bm{\mu^\star}) = \frac{d}{2}\log(2\pi\tau^2)+\log K+\frac{1}{2\tau^2}\Phi(\bm{\mu})+r_\tau(\bm{\mu}),        \label{eq:Ltau_expansion_pointwise}
    \end{align}
    where
    \begin{align}
        r_\tau(\bm{\mu}) \triangleq -\mathbb{E}_{Y\sim p_{\bm{\mu}^\star,\sigma}} \left[ \log\left(\sum_{\ell\in[K]}\exp\left(-\frac{\|Y-\mu_\ell\|^2-\min_{j\in[K]}\|Y-\mu_j\|^2}{2\tau^2} \right) \right) \right].
        \label{eq:rtau_def}
    \end{align}

    \item The remainder satisfies the uniform bounds $-\log K\le r_\tau(\bm{\mu})\le 0$ for all $\bm{\mu}\in\mathcal{U}$ and all $\tau>0$, and for every fixed $\bm{\mu}\in\mathcal{U}$,
    \begin{align}
        r_\tau(\bm{\mu}) \xrightarrow[\tau\to0^+]{} 0.
        \label{eq:rtau_to_0}
    \end{align}

    \item If $\bm{\mu}_\tau^\star\in\argmin_{\bm{\mu}\in\mathcal{U}}\mathcal{L}_\tau(\bm{\mu}; \bm{\mu^\star})$ and $\bm{\mu}_{\mathrm{HA}}^\star\in\argmin_{\bm{\mu}\in\mathcal{U}}\Phi(\bm{\mu})$ are the corresponding minimizers, then
    \begin{align}
        \lim_{\tau\to0^+} d_{\mathrm{perm}}\big(\bm{\mu}_\tau^\star,\bm{\mu}_{\mathrm{HA}}^\star\big)=0.  \label{eq:mu_tau_to_mu_HA}
    \end{align}
\end{enumerate}

\end{proposition}

\section{Component mean estimation under variance misspecification}
\label{sec:main-results}

We now develop the theory for estimating the component means under variance misspecification. To parameterize the mismatch level in a scale-free way, recall the mismatch ratio $\rho=\tau/\sigma$ from~\eqref{eq:rho-def}. It defines three regimes: under-smoothed fitting ($\rho<1$), correct specification ($\rho=1$; see Section~\ref{subsec:likelihood_and_kmeans}), and over-smoothed fitting ($\rho>1$). In this section, we analyze the over-smoothed regime in Section~\ref{sec:oversmoothed}, the under-smoothed regime in Section~\ref{sec:undersmoothed}, and conclude with a finite-sample interpretation in Section~\ref{sec:finite-sample-analysis}.

%Figure~\ref{fig:3} summarizes this picture in the $(\mathrm{SNR},\rho^2)$ plane. 

\subsection{Over-smoothed regime}
\label{sec:oversmoothed}

We begin with the \emph{over-smoothed} regime, in which the fitted variance exceeds the true one, namely,  $\rho=\tau/\sigma>1$. In this regime, the fitted likelihood treats the observations as noisier than they actually are and therefore tends to blur the cluster structure. As a result, the population objective leads to collapsed configurations in which several estimated means merge.

A canonical collapsed configuration is
\begin{align}
    \bm{\mu}_{\mathrm{coll}} \triangleq (\bar{\mu}^\star,\dots,\bar{\mu}^\star),
    \qquad
    \bar{\mu}^\star \triangleq \frac{1}{K}\sum_{\ell\in[K]} \mu_\ell^\star,
\end{align}
that is, the configuration in which all estimated means coincide with the global mixture center. %This behavior is qualitatively different from under-smoothing: rather than favoring sharper separation, over-smoothing suppresses effective cluster separation and stabilizes merged solutions.
For a general number of components $K$, the collapse threshold is governed by the covariance matrix of the true means configuration,
\begin{align}
    \Sigma_\mu \triangleq \frac{1}{K}\sum_{\ell\in[K]} (\mu_\ell^\star-\bar{\mu}^\star)(\mu_\ell^\star-\bar{\mu}^\star)^\top,
    \label{eqn:centered-second-moment-matrix}
\end{align}
whose largest eigenvalue, denoted by $\lambda_{\max}(\Sigma_\mu)$, identifies the direction of maximal variation between clusters.

\begin{proposition}[Collapse threshold for general $K$]
\label{prop:collapse_generalK}
Assume an equal-weight isotropic GMM as in Model~\ref{model:gmm_isotropic_equal}, with true means $\{\mu_\ell^\star\}_{\ell\in[K]}\subset\mathbb{R}^d$ and noise variance $\sigma^2 I_d$. Let $\bar{\mu}^\star\triangleq \frac{1}{K}\sum_{\ell\in[K]}\mu_\ell^\star$, and let $\Sigma_\mu$ be defined by~\eqref{eqn:centered-second-moment-matrix}. Then the following hold:
\begin{enumerate}
    \item The collapsed configuration $\bm{\mu}_{\mathrm{coll}}=(\bar{\mu}^\star,\dots,\bar{\mu}^\star)$ is a stationary point of $\mathcal{L}_\tau$, and it is a local minimum if and only if
    \begin{align}
        \rho^2 \ge 1+\frac{\lambda_{\max}(\Sigma_\mu)}{\sigma^2}.
    \end{align}

    \item  By~\eqref{eq:snr-def}, $\tr(\Sigma_\mu)=\sigma^2\,\mathrm{SNR}$, and thus the collapse threshold satisfies
    \begin{align}
        \label{eqn:lambda_max_bounds}
        1+\frac{\mathrm{SNR}}{d} \le 1+\frac{\lambda_{\max}(\Sigma_\mu)}{\sigma^2} \le 1+\mathrm{SNR}.
    \end{align}
\end{enumerate}
\end{proposition}

Proposition~\ref{prop:collapse_generalK}, which is proved in Appendix~\ref{sec:proof_collapse_generalK}, identifies the exact onset of collapse. Below the threshold, the merged configuration is unstable and the population objective favors separated component means. Above the threshold, the collapsed configuration becomes locally stable, showing that sufficiently strong over-smoothing can destroy local identifiability already at the population level. The threshold is governed by the geometry of the true mean configuration through $\Sigma_\mu$: it depends on the most energetic between-cluster direction, $\lambda_{\max}(\Sigma_\mu)$. The bounds in~\eqref{eqn:lambda_max_bounds} place the transition between two natural SNR scales.

The next two corollaries, proved in Appendices~\ref{sec:proof_collapse_K2}-\ref{sec:proof_collapse_simplex} illustrate the two extremal cases of~\eqref{eqn:lambda_max_bounds}. The symmetric two-component model attains the upper bound, while the regular-simplex geometry attains the lower bound.

\begin{corollary}[Symmetric two-component case]
\label{cor:collapse_K2}
Assume $K=2$ and $(\mu_0^\star,\mu_1^\star)=(\mu,-\mu)$. Then, the collapsed configuration $(0,0)$ is a local minimum of $\mathcal{L}_\tau$ if and only if
\begin{align}
    \rho^2 \ge 1+\mathrm{SNR}.
\end{align}
\end{corollary}

\begin{definition}[Regular simplex in $\mathbb{R}^d$]
\label{def:regular_simplex}
Fix integers $K\ge 3$ and $d\ge K-1$. A collection of vectors $\bm{v}=(v_0,\dots,v_{K-1})\in(\mathbb{R}^d)^K$ is called a \emph{regular simplex} if
\begin{align}
    \|v_\ell\|=1,
    \qquad
    \langle v_\ell,v_m\rangle=-\frac{1}{K-1}\ \ (\ell\neq m),
    \qquad
    \sum_{\ell=0}^{K-1} v_\ell=0.
    \label{eq:regular_simplex_def}
\end{align}
\end{definition}

\begin{corollary}[Regular simplex case]
\label{cor:collapse_simplex}
Assume the true means form a scaled regular simplex, namely, $\mu_\ell^\star=\beta v_\ell$, where $\{v_\ell\}_{\ell\in[K]}$ is a regular simplex in the sense of Definition~\ref{def:regular_simplex}. Then $\lambda_{\max}(\Sigma_\mu)=\beta^2/(K-1)$, and therefore the collapsed configuration is a local minimum of $\mathcal{L}_\tau$ if and only if
\begin{align}
    \rho^2 \ge 1+\frac{\mathrm{SNR}}{K-1}.
\end{align}
\end{corollary}

\subsection{Under-smoothed regime}
\label{sec:undersmoothed}

We now turn to the under-smoothed regime, in which the fitted variance is smaller than the true one, namely $\rho\triangleq\tau/\sigma<1$. This regime interpolates between the correctly specified model at $\rho=1$ and the hard-assignment limit as $\tau\to0$, and is therefore the natural setting in which to quantify how hard-assignment behavior emerges as the fitted temperature is decreased. Empirically, Figure~\ref{fig:3} indicates that the population error already approaches its hard-assignment behavior for moderately small values of $\rho<1$, well before $\tau$ becomes vanishingly small. Our goal in this subsection is to understand how the population quasi-MLE $\bm{\mu}_\tau^\star$ deviates from the ground-truth means in the low- and high-SNR regimes.

\subsubsection{Low-SNR regime}
We begin with the low-SNR regime. The next theorem, which is proved in Appendix~\ref{sec:proof_lowSNR_under_smoothed_fixed_rho}, shows that under-smoothing induces a genuine \emph{population drift}: for every fixed $\rho<1$, the population minimizer cannot remain close to the ground truth as the noise level grows. Thus, the effect of under-smoothing is not merely algorithmic or finite-sample; rather, even with infinite data, the fitted objective targets the wrong parameter.

\begin{thm}[Low-SNR under-smoothed mismatch for fixed $\rho<1$]
\label{thm:lowSNR_under_smoothed_fixed_rho}
Fix integers $K\ge 2$ and $d\ge 1$, and let the ground-truth means $\bm{\mu}^\star=(\mu_\ell^\star)_{\ell\in[K]}\in(\mathbb{R}^d)^K$ be distinct. Fix a variance mismatch ratio $\rho:=  \tau / \sigma \in(0,1)$, and for each $\sigma>0$ let $\bm{\mu}_{\tau}^\star\in\argmin_{\bm{\mu}\in(\mathbb{R}^d)^K}\mathcal{L}_{\tau}(\bm{\mu})$, where $\mathcal{L}_{\tau}$ is the population variance-mismatched negative log-likelihood defined in~\eqref{eq:def_L_tau}. Then,  there exist constants $C_\rho > 0$ and $\sigma_0<\infty$, depending on $K$, $d$, $\bm{\mu}^\star$, and $\rho$ but not on $\sigma$, such that
\begin{align}
    d_{\mathrm{perm}}^{2}\!\left(\bm{\mu}_{\tau}^\star,\bm{\mu}^\star\right) \ge C_\rho\,\sigma^{2},
    \qquad \forall\,\sigma\ge \sigma_0.
\end{align}
\end{thm}

\begin{remark}
    The constant $C_\rho$ in Theorem~\ref{thm:lowSNR_under_smoothed_fixed_rho} necessarily depends on $\rho$. Indeed, at $\rho=1$ the model is correctly specified, and the population minimizer coincides with the ground-truth means, so one must have $C_\rho \to 0$ as $\rho \to 1$. More generally, it is natural to expect that the under-smoothing bias is monotone in $\rho$ (i.e., the bias increases as $\rho$ decreases on $(0,1]$), although we do not pursue such a refinement here. 
\end{remark}

Theorem~\ref{thm:lowSNR_under_smoothed_fixed_rho} shows that, for every fixed $\rho<1$, under-smoothing induces an intrinsic low-SNR population bias of order at least $\sigma^2$. Since hard assignment corresponds to the limit $\tau\to0$, its low-SNR failure is part of this broader under-smoothed phenomenon.
The next corollary records the hard-assignment limit of this result.

\begin{corollary}[Low-SNR hard-assignment mismatch]
\label{cor:lowSNR_HA_mismatch}
Assume Assumption~\ref{ass:unique_minimizers_up_to_perm}. Then, there exist constants $C>0$  and $\sigma_0<\infty$, depending on $K$, $d$, $\bm{\mu}^\star$ but not on $\sigma$, such that
\begin{align}
    d_{\mathrm{perm}}^{2}\!\left(\bm{\mu}_{\mathrm{HA}}^\star(\sigma),\bm{\mu}^\star\right) \ge C\,\sigma^{2},
    \qquad \forall\,\sigma\ge \sigma_0.
\end{align}
\end{corollary}

Corollary~\ref{cor:lowSNR_HA_mismatch} shows that in the low-SNR regime, hard assignment is not merely statistically inefficient; rather, its population target itself drifts away from the ground truth. Thus, even with infinite data, hard assignment does not act as a consistent estimator of the true means in this regime.

\subsubsection{High-SNR regime}

In the high-SNR regime, the picture is very different. Unlike the low-SNR setting, where under-smoothing can substantially displace the population target, its effect becomes much weaker when the component means are well separated relative to the noise level. In that regime, observations rarely fall closer to the wrong component, and such misassignments occur with exponentially small probability. One therefore expects the entire under-smoothed family to remain uniformly close to the ground truth. The next proposition, proved in Appendix~\ref{sec:proof-prop-highSNR_under_smoothed_upper_uniform}, makes this high-SNR stability precise.

\begin{proposition}[High-SNR upper bound in the under-smoothed regime]
\label{prop:highSNR_under_smoothed_upper_uniform}
Assume Assumption~\ref{ass:unique_minimizers_up_to_perm}. Recall the definition of $\Delta_{\min}$ from~\eqref{eq:delta_min_def} and define $\Delta_{\max} \triangleq \max_{\ell\neq j}\|\mu_\ell^\star-\mu_j^\star\|$.
For each $\sigma>0$, let  $\tau\in(0,\sigma]$.
Then, for every $\eta>0$, there exists $\sigma_0=\sigma_0(\eta)>0$, depending only on $K,d,\bm{\mu}^\star,\eta$, such that for all $\sigma\le \sigma_0$,
\begin{align}
    \sup_{0<\tau\le \sigma} d_{\mathrm{perm}}^2  \bigl(\bm{\mu}_{\tau}^\star,\bm{\mu}^\star\bigr) \le 4K^2(K-1)\bigl(\Delta_{\max}^2+\sigma^2 d\bigr) \exp\!\left( -\Bigl(\frac{1}{8}-\eta\Bigr)\frac{\Delta_{\min}^2}{\sigma^2} \right).
\label{eq:highSNR_under_smoothed}
\end{align}
In particular, 
\begin{align}
    \sup_{0<\tau\le \sigma} d_{\mathrm{perm}} \bigl(\bm{\mu}_{\tau}^\star,\bm{\mu}^\star\bigr) \xrightarrow[\sigma\to0]{}0. \label{eq:high_SNR_conv}
\end{align}
\end{proposition}

Proposition~\ref{prop:highSNR_under_smoothed_upper_uniform} shows that, in the well-separated regime, variance under-smoothing does not destroy consistency at the population level: uniformly over $0<\rho\le 1$, the population minimizer remains exponentially close to the true means. Thus, the population drift established in Theorem~\ref{thm:lowSNR_under_smoothed_fixed_rho} is a genuinely low-SNR phenomenon.
At the hard-assignment limit $\tau\to 0$, this conclusion remains true. The next corollary, proved in Appendix~\ref{subsec:proof_highSNR_HA}, makes this precise.

\begin{corollary}[High-SNR exponential consistency of the population hard-assignment estimator]
\label{prop:highSNR_mean_upper}
Assume Assumption~\ref{ass:unique_minimizers_up_to_perm}. 
Then, for every $\eta>0$, there exists $\sigma_0=\sigma_0(\eta)>0$, depending only on $K,d,\bm{\mu}^\star,\eta$, such that for all $\sigma\le \sigma_0$,
\begin{align}
    d_{\mathrm{perm}}^2\bigl(\bm{\mu}_{\mathrm{HA}}^\star(\sigma),\bm{\mu}^\star\bigr) \le 4K^2(K-1)\bigl(\Delta_{\max}^2+\sigma^2 d\bigr) \exp\!\left( -\Bigl(\frac{1}{8}-\eta\Bigr)\frac{\Delta_{\min}^2}{\sigma^2} \right).
\label{eq:highSNR_hard_assignment}
\end{align}
In particular,
\begin{align}
    d_{\mathrm{perm}}\bigl(\bm{\mu}_{\mathrm{HA}}^\star(\sigma),\bm{\mu}^\star\bigr)  \xrightarrow[\sigma\to0]{}0.
\end{align}
\end{corollary}

Taken together, Theorem~\ref{thm:lowSNR_under_smoothed_fixed_rho} and Proposition~\ref{prop:highSNR_under_smoothed_upper_uniform} show that the under-smoothed regime exhibits a sharp SNR-dependent dichotomy: in low SNR, the population target drifts away from the truth, whereas in high SNR it remains exponentially close to it. Corollary~\ref{cor:lowSNR_HA_mismatch} and Corollary~\ref{prop:highSNR_mean_upper} show that the same dichotomy persists at the hard-assignment endpoint. Thus, hard assignment is accurate in the well-separated regime, where it effectively agrees with the true latent partition, but it becomes fundamentally biased in low SNR.

\subsection{Finite-sample analysis}
\label{sec:finite-sample-analysis}

We conclude this section with a brief finite-sample interpretation of the population theory developed above. The main point is to separate two distinct sources of estimation error: the \emph{finite-sample fluctuation} around the relevant population target, and the \emph{population mismatch} between that target and the ground-truth means. This viewpoint applies uniformly across the variance-mismatched family indexed by $\tau>0$, with hard assignment arising as $\tau\to 0$.

Fix $\tau>0$, and let $\widehat{\bm{\mu}}_{\tau}^{(n)}\in\argmin_{\bm{\mu}}\mathcal{L}_{\tau}^{(n)}(\bm{\mu};\mathcal{Y})$ denote an empirical variance-mismatched likelihood estimator, where $\mathcal{L}_{\tau}^{(n)}$ is the empirical negative log-likelihood defined in~\eqref{eq:empirical_loglik}. Let $\bm{\mu}_\tau^\star\in\argmin_{\bm{\mu}}\mathcal{L}_{\tau}(\bm{\mu})$ denote the corresponding population target from~\eqref{eqn:quasi-MLE}. Then, by the triangle inequality for the permutation-invariant distance $d_{\mathrm{perm}}$ from~\eqref{eq:perm-distance},
\begin{align}
    d_{\mathrm{perm}}(\widehat{\bm{\mu}}_{\tau}^{(n)},\bm{\mu}^\star) \le d_{\mathrm{perm}}(\widehat{\bm{\mu}}_{\tau}^{(n)},\bm{\mu}_\tau^\star) + d_{\mathrm{perm}}(\bm{\mu}_\tau^\star,\bm{\mu}^\star).
    \label{eq:finite-infinite-perm-distance}
\end{align}
Using $(a+b)^2\le 2a^2+2b^2$ together with the definition of the normalized MSE in~\eqref{eq:mseDef_gmm}, we obtain
\begin{align}
    \mathrm{MSE}(\widehat{\bm{\mu}}_{\tau}^{(n)},\bm{\mu}^\star) \le \frac{2}{\|\bm{\mu}^\star\|_F^2}\,   \mathbb{E}\!\left[d_{\mathrm{perm}}^2(\widehat{\bm{\mu}}_{\tau}^{(n)},\bm{\mu}_\tau^\star)\right] + \frac{2}{\|\bm{\mu}^\star\|_F^2}\, d_{\mathrm{perm}}^2(\bm{\mu}_\tau^\star,\bm{\mu}^\star).
    \label{eq:finite-sample-decomp-tau}
\end{align}
The first term is a finite-sample fluctuation term, while the second is a purely population-level mismatch term.

We do not pursue a full nonasymptotic finite-sample theory here. Rather, \eqref{eq:finite-sample-decomp-tau} serves as the basic asymptotic decomposition suggested by misspecified maximum-likelihood theory. Under standard regularity assumptions, the empirical minimizer $\widehat{\bm{\mu}}_{\tau}^{(n)}$ is consistent for the population minimizer $\bm{\mu}_\tau^\star$ of the misspecified objective and satisfies the asymptotic normality expansion around it; see, for example,~\cite[Thm.~2.2 and Sec.~3]{white1982maximum}. Consequently, $\mathrm{MSE}(\widehat{\bm{\mu}}_{\tau}^{(n)},\bm{\mu}_\tau^\star) = O\!\left(1/n\right)$, with an implicit constant determined by the local misspecified likelihood geometry at $\bm{\mu}_\tau^\star$.

The second term in~\eqref{eq:finite-sample-decomp-tau}, namely $d_{\mathrm{perm}}^2(\bm{\mu}_\tau^\star,\bm{\mu}^\star)$, is entirely a population quantity. It depends on the signal geometry, the noise level $\sigma$, and the mismatch ratio $\rho = \tau/\sigma$, but not on the sample size $n$. In the correctly specified case $\tau=\sigma$, the population target coincides with $\bm{\mu}^\star$ up to permutation, and this term vanishes. For $\tau\neq\sigma$, however, it is precisely the population drift analyzed in the previous subsections. In particular, it grows on the order of $\sigma^2$ in low SNR for any $\rho<1$.

Thus, combining these two terms into~\eqref{eq:finite-sample-decomp-tau} shows that the total estimation error is controlled by the sum of a fluctuation term and a population mismatch term. In asymptotic form, this yields
\begin{align}
    \mathrm{MSE}(\widehat{\bm{\mu}}_{\tau}^{(n)},\bm{\mu}^\star) \lesssim \frac{A_\tau}{n} + B_\tau,    \label{eq:bias_variance_crossover}
\end{align}
where $A_\tau$ summarizes the local fluctuation scale and $B_\tau \triangleq \frac{2}{\|\bm{\mu}^\star\|_F^2}\, d_{\mathrm{perm}}^2(\bm{\mu}_\tau^\star,\bm{\mu}^\star)$ is the population mismatch floor.

This decomposition explains the two qualitatively different behaviors seen in Figure~\ref{fig:2}(b). When $B_\tau$ is negligible, as in the correctly specified model or in the high-SNR regime where the population drift is small, the dominant contribution to the error is the $O(1/n)$ fluctuation term, so increasing the sample size continues to reduce the MSE. By contrast, when $B_\tau$ is appreciable, as in the low-SNR regime, the $n$-independent population mismatch eventually dominates. More precisely, the crossover occurs at a sample size of order $n_{\mathrm{cross}} \asymp \frac{A_\tau}{B_\tau}$. Accordingly, for $n \ll n_{\mathrm{cross}}$ the error is fluctuation-dominated, whereas for $n \gg n_{\mathrm{cross}}$ it is mismatch-dominated. Beyond this point, additional samples reduce only the fluctuations around a biased population target, and the overall MSE saturates near the mismatch floor. In particular, the finite-sample analysis is most substantial in the well-separated, high-SNR regime, where $B_\tau$ is small and the estimator exhibits the familiar $1/n$ improvement. In low SNR, by contrast, the main limitation is typically not finite-sample variability but the population bias induced by variance mismatch itself.

\section{Latent-label clustering}
\label{sec:latent-label-misclassification}

Although our main focus is component mean estimation, the hard-assignment viewpoint developed above also leads naturally to latent-label recovery. In a GMM, this task asks: given an observation $Y$, infer the component label $L\in[K]$ from which it was generated. This connection is especially relevant here because, in the equal-weight isotropic model, the Bayes-optimal classifier is a nearest-center rule. Hence, the same geometric partition of the observation space that governs optimal clustering also appears in the assignment step of hard-assignment procedures. As a result, our analysis of hard assignment has direct implications for clustering: when latent labels are statistically unrecoverable, any method that relies on confident per-sample assignments is necessarily fragile. This is the decision-theoretic counterpart of the population-drift phenomenon established earlier for under-smoothed fitting and for the hard-assignment limit.

The aim of this section is to make this connection precise. We first define the optimal achievable clustering error, then identify the Bayes-optimal classifier and relate it to the nearest-center geometry underlying hard assignment. We next show that in low SNR the Bayes misclassification probability is necessarily close to random guessing, whereas in high SNR it decays exponentially fast with the separation-to-noise ratio.

\begin{definition}[Bayes misclassification probability]
\label{def:latent_label_model_isotropic}
Let $(Y,L)$ be distributed according to Model~\ref{model:gmm_isotropic_equal}. A measurable classifier is a map $\psi:\mathbb{R}^d\to[K]$, producing the estimate $\widehat L=\psi(Y)$. The Bayes, or minimum achievable, misclassification probability is
\begin{align}
    P_{\mathrm{err}}^\star(\bm{\mu}^\star,\sigma) \triangleq \inf_{\psi:\mathbb{R}^d\to[K]} \mathbb{P}\big(\psi(Y)\neq L\big).
    \label{eq:bayes_error_def_kmeans}
\end{align}
\end{definition}

The next proposition, a standard consequence of Gaussian discriminant analysis with common covariance; see, for example,~\cite{bishop2006pattern}, shows that in the equal-weight isotropic model the Bayes-optimal classifier is exactly the nearest-center rule. Thus, latent-label recovery and hard assignment are governed by the same Voronoi partition of the observation space, even though they address different statistical tasks: the Bayes classifier characterizes the best possible label recovery under the true model, whereas hard assignment uses the same decision regions as part of a mean estimation procedure.

\begin{proposition}[Bayes classifier is nearest-center]
\label{prop:bayes_is_nearest_mean}
Consider Model~\ref{model:gmm_isotropic_equal} with uniform weights and common covariance $\sigma^2 I_d$. Then the Bayes-optimal classifier $\psi^\star$ is given by the nearest-center rule:
\begin{align}
    \psi^\star(y) &\in \argmax_{\ell\in[K]} \mathbb{P}(L=\ell\mid Y=y)
    \notag\\
    &=
    \argmax_{\ell\in[K]} \varphi_d\!\left(y;\mu_\ell^\star,\sigma^2 I_d\right) = \argmin_{\ell\in[K]} \|y-\mu_\ell^\star\|^2.
    \label{eq:bayes_is_nearest_mean}
\end{align}
\end{proposition}

\subsection{Low-SNR regime}

We begin with the low-SNR regime. When the observation $Y$ carries very little information about the latent label $L$, even the optimal Bayes classifier cannot substantially outperform random guessing. To formalize this, we relate the Bayes misclassification probability to the mutual information $I(L;Y)$ between $L$ and $Y$ under the true GMM law~\cite{cover1999elements}. 

\begin{proposition}[Bayes error is close to random guessing at low SNR]
\label{prop:low_snr_lb_classification}
Let $(Y,L)$ be distributed according to Model~\ref{model:gmm_isotropic_equal}, and let $P_{\mathrm{err}}^\star(\bm{\mu}^\star,\sigma)$ be defined as in~\eqref{eq:bayes_error_def_kmeans}. Then, 
\begin{align}
    P_{\mathrm{err}}^\star(\bm{\mu}^\star,\sigma) \ge 1-\frac{1}{K}-\sqrt{\frac{I(L;Y)}{2}}, \label{eq:bayes_error_lb_in_terms_of_MI}
\end{align}
where $I(L;Y)$ is the mutual information between $L$ and $Y$.
\end{proposition}

Proposition~\ref{prop:low_snr_lb_classification}, proved in Appendix~\ref{sec:proof-lowSNR_classification}, shows that any nontrivial improvement over random guessing requires the mutual information $I(L;Y)$ to be bounded away from zero. The next corollary bounds this mutual information explicitly in terms of the SNR.

\begin{corollary}[Low-SNR mutual-information bound]
\label{cor:upper-bound-mutual-information}
The mutual information in~\eqref{eq:bayes_error_lb_in_terms_of_MI} satisfies the nonasymptotic bound
\begin{align}
    I(L;Y) \le \frac{1}{2\sigma^2}\cdot \frac{1}{K}\sum_{\ell\in[K]}\|\mu_\ell^\star-\bar{\mu}^\star\|^2 = \frac{\mathrm{SNR}}{2},  \label{eq:MI_energy_upper_bound_raw}
\end{align}
where $\bar{\mu}^\star\triangleq \frac{1}{K}\sum_{\ell\in[K]}\mu_\ell^\star$ is the mixture mean. Therefore
\begin{align}
    P_{\mathrm{err}}^\star(\bm{\mu}^\star,\sigma) \ge 1-\frac{1}{K}-\frac{1}{2}\sqrt{\mathrm{SNR}}. \label{eq:bayes_error_lb_energy}
\end{align}
\end{corollary}

Corollary~\ref{cor:upper-bound-mutual-information}, proved in Appendix~\ref{sec:corr-upper-bound-mutual-information}, implies that whenever $\mathrm{SNR}\to 0$, the Bayes error approaches the random-guessing benchmark at a rate proportional to $\sigma^{-1}$:
\begin{align}
    P_{\mathrm{err}}^\star(\bm{\mu}^\star,\sigma) = 1-\frac{1}{K}-o(1),
    \qquad \sigma\to\infty.
    \label{eq:random_guessing_limit}
\end{align}
Thus, in the low-SNR regime, even the optimal classifier cannot reliably recover the latent labels from a single sample. This provides an information-theoretic explanation for the brittleness of hard-assignment methods in the same regime: if the labels themselves are nearly unrecoverable, then any procedure that relies on accurate hard per-sample decisions is necessarily unstable.

\subsection{High-SNR regime}

We next complement the low-SNR impossibility result with a high-SNR guarantee. When the component means are well separated relative to the noise level, the Bayes classifier is effectively a nearest-mean rule, and errors occur only when the noise pushes an observation across one of the pairwise decision hyperplanes. Consequently, the misclassification probability decays exponentially fast in the $\Delta_{\min}^2/\sigma^2$ ratio.
The next proposition makes this precise; its proof is given in Appendix~\ref{sec:proof-of-highSNR-classification}.

\begin{proposition}[High-SNR exponential upper bound on misclassification]
\label{prop:highSNR_exp_bound}
Recall the definition of the pairwise separations $\Delta_{\ell j}\triangleq \|\mu_\ell^\star-\mu_j^\star\|$, and $\Delta_{\min}\triangleq \min_{\ell\neq j}\Delta_{\ell j}$ from~\eqref{eq:delta_min_def}.
Then, the Bayes misclassification probability~\eqref{eq:bayes_error_def_kmeans} satisfies
\begin{align}
    P_{\mathrm{err}}^\star(\bm{\mu}^\star,\sigma) \le \frac{1}{2K}\sum_{\ell\in[K]}\sum_{j\in[K]\setminus\{\ell\}}\exp\!\left(-\frac{\Delta_{\ell j}^2}{8\sigma^2}\right).    \label{eq:highSNR_pairwise_sum_bound}
\end{align}
In particular,
\begin{align}
    P_{\mathrm{err}}^\star(\bm{\mu}^\star,\sigma) \le \frac{K-1}{2}\exp\!\left(-\frac{\Delta_{\min}^2}{8\sigma^2}\right).
    \label{eq:highSNR_minsep_bound}
\end{align}
\end{proposition}

\begin{remark}
\label{remark:sharper_asymptotics_main}
Proposition~\ref{prop:highSNR_exp_bound} applies for every noise level $\sigma$. If one is interested in sharper high-SNR asymptotics as $\sigma \to 0$, then the Gaussian tail bound used in the proof can be refined, yielding the same exponential rate as in~\eqref{eq:highSNR_minsep_bound} but with a sharper prefactor; in particular, this prefactor can be bounded uniformly by a term of order $\sigma/\Delta_{\min}$. See the corresponding Remark~\ref{rem:mills_ratio_high_snr} for details.
\end{remark}

Proposition~\ref{prop:highSNR_exp_bound} shows that latent-label recovery undergoes a sharp transition across SNR regimes. In the high-SNR regime, the Bayes error is exponentially small, and nearest-center decisions are reliable. In the low-SNR regime, by contrast, the Bayes error is essentially indistinguishable from random guessing. This dichotomy mirrors the behavior of hard-assignment mean estimation established in Section~\ref{sec:main-results}: the high-SNR estimate~\eqref{eq:highSNR_minsep_bound} is the clustering analogue of Corollary~\ref{prop:highSNR_mean_upper}, whereas the near-random-guessing regime in~\eqref{eq:random_guessing_limit} is the decision-theoretic counterpart of the low-SNR hard-assignment drift in Corollary~\ref{cor:lowSNR_HA_mismatch}. Thus, hard assignment is accurate when the Voronoi partition induced by the true means is reliable, but it becomes fundamentally fragile once that partition ceases to carry substantial label information.

\section{Symmetric two-component Gaussian mixtures}
\label{sec:theoretical-results-structured-GMM}
We now specialize the general theory to the two-component case. Consider first a general two-component configuration $(\mu_1,\mu_2)$ in Model~\ref{model:gmm_isotropic_equal} with $K=2$. Define its midpoint and half-difference by $m \triangleq \frac{1}{2}(\mu_1+\mu_2)$, and $\mu \triangleq \frac{1}{2}(\mu_1-\mu_2)$.
Then, the means may be written as $\mu_1=m+\mu$ and $\mu_2=m-\mu$. Since the population $k$-means objective is invariant under joint translation of the data and component means, and $d_{\mathrm{perm}}$ is likewise translation-invariant, any such two-component model can be reduced, without loss of generality, to the centered symmetric form obtained by subtracting $m$.
Indeed, with $Y'\triangleq Y-m$, the translated model becomes
\begin{align}
    Y' \sim \frac{1}{2}\,\mathcal{N}(\mu,\sigma^2 I_d) + \frac{1}{2}\,\mathcal{N}(-\mu,\sigma^2 I_d). \label{eq:K2_d_gmm_2}
\end{align}

We therefore work throughout this section with the isotropic symmetric two-component GMM~\eqref{eq:K2_d_gmm_2}, so that the ground-truth center configuration is $\bm{\mu}^\star=(\mu,-\mu)\in(\mathbb{R}^d)^2$.
This is the simplest nontrivial benchmark in which the two main phenomena identified above already appear: low-SNR failure of hard assignment and high-SNR recovery. At the same time, the model is sufficiently explicit to permit closed-form analysis, allowing the general qualitative picture of Sections~\ref{sec:main-results} and~\ref{sec:latent-label-misclassification} to be sharpened into exact formulas.

Since the mixture is centered, $\bar{\mu}^\star=0$, the SNR defined
in~\eqref{eq:snr-def} reduces to
\begin{align}
    \mathrm{SNR} = \frac{1}{\sigma^2}\cdot \frac{1}{2}\bigl(\|\mu\|^2+\|-\mu\|^2\bigr) = \frac{\|\mu\|^2}{\sigma^2}.
    \label{eq:SNR_K2_identity}
\end{align}
We begin with the component mean estimation error of the population hard-assignment. In the symmetric $K=2$ model, its permutation-invariant error admits an explicit closed form, as stated next.

\begin{thm}[Population hard-assignment error for symmetric $K=2$]
\label{thm:perm_MSE_d_K2}
Fix $d\ge 1$ and consider the symmetric two-component model~\eqref{eq:K2_d_gmm_2}. Let $\bm{\mu}_{\mathrm{HA}}^\star(\sigma)$ denote the population hard-assignment minimizer introduced in Section~\ref{sec:main-results}. Then, 
\begin{align}
    \frac{d_{\mathrm{perm}}^2\!\big(\bm{\mu}_{\mathrm{HA}}^\star(\sigma),\bm{\mu}^\star\big)}{\|\bm{\mu}^\star\|_F^2} = \frac{1}{\pi\,\|\mu\|^2} \left(\sqrt{2}\,\sigma\,\exp\!\Big(-\frac{\|\mu\|^2}{2\sigma^2}\Big) -  \operatorname{erfc}\!\Big(\frac{\|\mu\|}{\sqrt{2}\sigma}\Big)\,\|\mu\| \right)^2.
    \label{eq:K2_d_permMSE_normalized}
\end{align}
\end{thm}

Theorem~\ref{thm:perm_MSE_d_K2}, proved in Appendix \ref{app:MSEForK2}, gives the exact population error of hard assignment in the symmetric two-component model. In particular, the error depends only on the SNR in~\eqref{eq:SNR_K2_identity}. The next corollary records the resulting low- and high-SNR asymptotics.

\begin{corollary}[Low- and high-SNR asymptotics for the population hard-assignment error]
\label{cor:K2_HA_asymptotics}
In the setting of Theorem~\ref{thm:perm_MSE_d_K2}, the population hard-assignment error satisfies the following asymptotics:
\begin{align}
    \mathrm{MSE}\big(\bm{\mu}_{\mathrm{HA}}^\star(\sigma),\bm{\mu}^\star\big)
    =
    \begin{cases}
        \dfrac{2}{\pi }\,\dfrac{1}{\mathrm{SNR}}\bigl(1+o(1)\bigr),
        & \mathrm{SNR}\to0, \\[3ex]
        \dfrac{2\exp\{-\mathrm{SNR}\}}{\pi}\,\dfrac{1}{(\mathrm{SNR})^3} \left(1+O\!\left(\dfrac{1}{\mathrm{SNR}}\right)\right),
        & \mathrm{SNR}\to\infty.
    \end{cases}
\end{align}
\end{corollary}

Corollary~\ref{cor:K2_HA_asymptotics}, proved in Appendix~\ref{app:AsymptoticSigmaInfK2}, makes the dichotomy from Section~\ref{sec:main-results} completely explicit in the symmetric two-component model. In low SNR, the population hard-assignment error grows on the order of $\mathrm{SNR}^{-1}$, in agreement with the general lower-bound phenomenon from Corollary~\ref{cor:lowSNR_HA_mismatch}. In high SNR, the error is exponentially small, consistent with Corollary~\ref{prop:highSNR_mean_upper}, but here the symmetric structure yields an exact closed form and sharp asymptotics. 

Next, we specialize the latent-label clustering analysis from Section~\ref{sec:latent-label-misclassification} into the same symmetric $K=2$ model. In this setting, the Bayes misclassification probability also admits a closed form. The proof of Proposition~\ref{prop:bayes_error_K2} is deferred to Appendix~\ref{sec:proof_bayes_error_K2}.

\begin{proposition}[Bayes misclassification error for symmetric $K=2$]
\label{prop:bayes_error_K2}
Fix $d\ge 1$ and consider the model~\eqref{eq:K2_d_gmm_2}. Then the Bayes misclassification probability~\eqref{eq:bayes_error_def_kmeans} is
\begin{align}
    P_{\mathrm{err}}^\star(\bm{\mu}^\star,\sigma) = \frac{1}{2}\,\operatorname{erfc}\!\left(\frac{\|\mu\|}{\sqrt{2}\,\sigma}\right) = \frac{1}{2}\,\operatorname{erfc}\!\left(\sqrt{\frac{\mathrm{SNR}}{2}}\right). \label{eq:bayes_error_K2_closed_form}
\end{align}
\end{proposition}

The next corollary, which is proved in Appendix~\ref{sec:proof_bayes_error_K2_asymptotics}, records the corresponding low- and high-SNR asymptotics.

\begin{corollary}[Low- and high-SNR asymptotics for the symmetric $K=2$ Bayes error]
\label{cor:bayes_error_K2_asymptotics}
In the setting of Proposition~\ref{prop:bayes_error_K2},
\begin{align}
    P_{\mathrm{err}}^\star(\bm{\mu}^\star,\sigma)
    =
    \begin{cases}
        \dfrac{1}{2} -\sqrt{\dfrac{\mathrm{SNR}}{2\pi}} +        O\!\bigl((\mathrm{SNR})^{3/2}\bigr),
        & \mathrm{SNR}\to0, \\[3ex]
        \dfrac{1}{\sqrt{2\pi \mathrm{SNR}}}\,
        \exp\!\left(-\dfrac{\mathrm{SNR}}{2}\right)\bigl(1+o(1)\bigr),
        & \mathrm{SNR}\to\infty.
    \end{cases}    \label{eq:bayes_error_K2_lowSNR}
\end{align}
\end{corollary}

\paragraph{Relation to the general bounds.}
The exact formulas above refine the general results of Sections~\ref{sec:main-results} and~\ref{sec:latent-label-misclassification}. For mean estimation, Corollary~\ref{cor:K2_HA_asymptotics} shows that the general low-SNR hard-assignment drift from Corollary~\ref{cor:lowSNR_HA_mismatch} is sharp in its $\sigma^2$-scaling, and Corollary~\ref{prop:highSNR_mean_upper} captures the correct exponential high-SNR decay.

For latent-label recovery, the general information-theoretic lower bound from Proposition~\ref{prop:low_snr_lb_classification} together with Corollary~\ref{cor:upper-bound-mutual-information} yields, for $K=2$,
\begin{align}
    P_{\mathrm{err}}^\star(\bm{\mu}^\star,\sigma) \ge \frac{1}{2} - \frac{1}{2}\sqrt{\mathrm{SNR}},
    \qquad \mathrm{SNR}\to 0,    \label{eq:bayes_error_general_lowSNR_specialized_K2}
\end{align}
which captures the correct $\Theta(\sqrt{\mathrm{SNR}})$ improvement over random guessing, but not the sharp constant. The exact expansion~\eqref{eq:bayes_error_K2_lowSNR} shows that the optimal constant is $1/\sqrt{2\pi}$. Likewise, Proposition~\ref{prop:highSNR_exp_bound} with $\Delta_{\min}=2\|\mu\|$ gives
\begin{align}
    P_{\mathrm{err}}^\star(\bm{\mu}^\star,\sigma) \le \frac{1}{2}\exp\!\left(-\frac{\|\mu\|^2}{2\sigma^2}\right) = \frac{1}{2}\exp\!\left(-\frac{\mathrm{SNR}}{2}\right),    \label{eq:bayes_error_general_highSNR_specialized_K2}
\end{align}
which has the correct exponential rate. The exact asymptotic~\eqref{eq:bayes_error_K2_lowSNR} further identifies the precise polynomial prefactor $(2\pi\,\mathrm{SNR})^{-1/2}$, which coincides with the asymptotic prefactor described in Remark~\ref{remark:sharper_asymptotics_main}.

\section{Discussion and outlook}
\label{sec:outlook}

The main conclusion of this work is that Gaussian mixture mean estimation can be highly sensitive to the variance parameter used in the fitted likelihood, with the effect becoming especially severe in the low-SNR regime. Variance misspecification is therefore not a secondary modeling issue: depending on the mismatch ratio $\rho=\tau/\sigma$, it can leave the target essentially unchanged, displace it away from the truth, or induce collapsed configurations. This is especially relevant in applications where the noise level must itself be estimated from the data, such as cryo-EM~\cite{lyumkis2019challenges,bendory2020single,scheres2012relion}, since even small calibration errors can move the fitted procedure into a qualitatively different regime. More broadly, in structural biology and low-SNR imaging, including cryo-EM and cryo-ET workflows~\cite{chen2019complete,schaffer2019cryo}, likelihood-based methods such as expectation-maximization are used throughout the analysis pipeline, from 2D classification to downstream estimation of class representatives and 3D volumes. Our results suggest that in such low-SNR settings, variance misspecification can induce systematic bias already at the population level, and not merely through finite-sample fluctuations.

A second conclusion is that the population geometry studied here has a direct counterpart in latent-label recovery and hard-assignment estimation, with implications for practical pipelines that rely on hard clustering or hard-assignment mean updates. While hard assignments are often introduced for computational simplicity, treating them as reliable estimates of latent component membership can be fundamentally misleading in low SNR. When latent labels are intrinsically difficult to recover, hard decisions at the level of individual observations can induce a biased population target. This is relevant to procedures in which samples are filtered, reassigned, or clustered using discrete labels, as in filtration stages that often follow 2D classification in cryo-EM~\cite{scheres2012relion,weiss2023unsupervised}. Consequently, the resulting bias need not be a finite-sample or optimization artifact; even an idealized infinite-data version of such a hard-assignment procedure may converge to a biased solution.

\paragraph{Future work.}
Several directions for future work are especially natural. First, it would be important to extend the present framework beyond finite discrete mixtures to continuous latent-variable models. This is particularly compelling for structural-biology applications such as cryo-EM and cryo-ET, where the latent variable is typically a continuous pose and the observation model includes projection and imaging operators in addition to noise. Extending the variance-mismatch perspective to such continuous-group models could clarify when pose estimation, hard assignment, or related approximation schemes are statistically justified, and when they induce a population bias analogous to the one identified here.

Second, the present work is primarily population-level and estimator-level. A natural next step is therefore to study the algorithmic behavior of expectation-maximization itself under variance misspecification~\cite{dwivedi2018theoretical}. In particular, it would be valuable to understand when expectation-maximization converges to the population quasi-maximum-likelihood target, when it is attracted to biased local optima, and how this behavior compares with the correctly specified setting, where expectation-maximization is often used to approach the maximum-likelihood estimator. Such results would connect the geometric picture developed here with the practical behavior of one of the most widely used algorithms for latent-variable inference.

\section*{Data Availability}
The detailed implementation and code are available at \href{https://github.com/VladiSer/GMM-Variance-Misspecification.git}{https://github.com/VladiSer/GMM-Variance-Misspecification}.

\section*{Acknowledgment}
T.B. is supported in part by BSF under Grant 2020159, in part by NSF-BSF under Grant 2019752, and in part by ISF under Grant 1924/21. 

\bibliographystyle{plain}
%\bibliography{refs}

\begin{appendices}

{\centering{\section*{Appendix}}}

\section{Preliminaries}

\subsection{Hessian at the origin  configuration}

Here we analyze the population objective at the origin configuration $\bm{\mu}=\bm{0}$. Throughout, we assume that the true means are centered, namely,  $\bar{\mu}^\star=0$ (equivalently, $\mathbb{E}[Y]=0$).

\begin{remark}
\label{remark:arbitrary_tie_breaking}
Throughout the proofs, nearest-center assignments of the form $\argmin_{m\in[K]}\|y-\mu_m\|^2$ may not be uniquely defined on the Voronoi boundaries $\{y:\|y-\mu_j\|=\|y-\mu_k\|\}$. Under our isotropic Gaussian model, and under distinct means $\bm{\mu}$, these boundaries have Lebesgue measure zero. We therefore adopt an arbitrary measurable tie-breaking rule wherever such assignments appear, and all subsequent statements hold for any such choice.
\end{remark}

\begin{lem}[Hessian of $\mathcal{L}_\tau$ at $\bm{\mu}=\bm{0}$]
\label{lem:hessian_blocks}
Let $Y$ be drawn according to model~\ref{model:gmm_isotropic_equal}, and recall the definition of $\mathcal{L}_\tau(\bm{\mu};\bm{\mu}^\star) $ from~\eqref{eq:def_L_tau}. At $\bm{\mu}=\bm{0}$, the Hessian $\nabla^2_{\bm{\mu}}\mathcal{L}_\tau(\bm{0};\bm{\mu}^\star)\in\mathbb{R}^{Kd\times Kd}$ consists of $K\times K$ blocks of size $d\times d$. Its diagonal and off-diagonal blocks are
\begin{align}
    \label{eq:Diag_L_Hess_term}  \nabla^2_{\bm{\mu}}\mathcal{L}_\tau(\bm{0};\bm{\mu}^\star)_{kk} &=
    \frac{I_d}{K\tau^2} -\frac{(K-1)\,\mathbb{E}[YY^\top]}{K^2\tau^4},
    \\    \label{eq:Off_Diag_L_Hess_term}    \nabla^2_{\bm{\mu}}\mathcal{L}_\tau(\bm{0};\bm{\mu}^\star)_{kj}
    &=
    \frac{\mathbb{E}[YY^\top]}{K^2\tau^4},
    \qquad k\neq j.
\end{align}
\end{lem}

\begin{proof}[Proof of Lemma~\ref{lem:hessian_blocks}]
By definition~\eqref{eq:def_L_tau}, $\mathcal{L}_\tau(\bm{\mu};\bm{\mu}^\star) = -\mathbb{E}_{Y\sim p_{\bm{\mu}^\star,\sigma}} \big[\log p_{\bm{\mu},\tau}(Y)\big]$, where the expectation is taken under the true distribution of $Y$. Differentiating under the expectation yields
\begin{align}
    \nabla^2_{\bm{\mu}}\mathcal{L}_\tau(\bm{\mu};\bm{\mu}^\star) = \mathbb{E}_{Y\sim p_{\bm{\mu}^\star,\sigma}}\!\left[\frac{(\nabla_{\bm{\mu}}p_{\bm{\mu},\tau}(Y))(\nabla_{\bm{\mu}}p_{\bm{\mu},\tau}(Y))^\top} {p_{\bm{\mu},\tau}(Y)^2} - \frac{\nabla^2_{\bm{\mu}}p_{\bm{\mu},\tau}(Y)} {p_{\bm{\mu},\tau}(Y)} \right].
\end{align}
We now evaluate this at $\bm{\mu}=\bm{0}$, where all estimated components coincide.

For each $k\in[K]$,
\begin{align}
    \nabla_{\mu_k} p_{\bm{\mu},\tau}(Y)\big|_{\bm{\mu}=\bm{0}} = \frac{1}{K\tau^2}\,Y\,\varphi_d(Y;0,\tau^2 I_d),
\end{align}
so
\begin{align}
    \left[\frac{(\nabla_{\mu_k}p_{\bm{\mu},\tau}(Y))(\nabla_{\mu_j}p_{\bm{\mu},\tau}(Y))^\top} {p_{\bm{\mu},\tau}(Y)^2} \right]_{\bm{\mu}=\bm{0}} = \frac{YY^\top}{K^2\tau^4}.
\end{align}
For the second derivatives, when $k=j$,
\begin{align}
    \left[-\frac{\nabla^2_{\mu_k}p_{\bm{\mu},\tau}(Y)} {p_{\bm{\mu},\tau}(Y)} \right]_{\bm{\mu}=\bm{0}} = \frac{I_d}{K\tau^2} - \frac{YY^\top}{K\tau^4},
\end{align}
whereas for $k\neq j$,
$\nabla_{\mu_k}\nabla_{\mu_j}p_{\bm{\mu},\tau}(Y)=0$.
Combining these contributions and taking expectations gives \eqref{eq:Diag_L_Hess_term} and \eqref{eq:Off_Diag_L_Hess_term}.
\end{proof}

As an immediate consequence, the origin configuration $\bm{\mu}=\bm{0}$ exhibits a direction of strict negative curvature whenever $\tau<1$. We will use this below to show that $\bm{\mu}=\bm{0}$ cannot be a local minimizer in the under-smoothed regime.

\begin{lem}[Negative curvature at the origin configuration in the zero-signal model]
\label{lem:negative_curvature_zero_signal}
Fix $\tau\in(0,1)$, and consider the zero-signal model $Y\sim\mathcal{N}(0,I_d)$, that is the true means are $\bm{\mu^\star} = 0$.
Then, for every $\bm{h}=(h_\ell)_{\ell=1}^K\in(\mathbb{R}^d)^K$ satisfying $\sum_{\ell=1}^K h_\ell=0$ and $\bm{h}\neq\bm{0}$,
\begin{align}
    \left.\frac{d^2}{dt^2}\mathcal{L}_\tau(t\bm{h};\bm{0})\right|_{t=0} = -\frac{1-\tau^2}{K\tau^4}\sum_{\ell=1}^K \|h_\ell\|^2 <0.
\end{align}
In particular, $\bm{0}$ is not a local minimizer of $\mathcal{L}_\tau(\,\cdot\,;\bm{0})$.
\end{lem}

\begin{proof}[Proof of Lemma~\ref{lem:negative_curvature_zero_signal}]
In the zero-signal model, $\mathbb{E}[YY^\top]=I_d$. Hence, by Lemma~\ref{lem:hessian_blocks}, the Hessian of $\mathcal{L}_\tau(\,\cdot\,;\bm{0})$ at $\bm{\mu}=\bm{0}$ has diagonal blocks
\begin{align}
    \nabla^2_{\bm{\mu}}\mathcal{L}_\tau(\bm{0};\bm{0})_{kk} = \frac{I_d}{K\tau^2} - \frac{K-1}{K^2\tau^4}I_d,
\end{align}
and off-diagonal blocks
\begin{align}
    \nabla^2_{\bm{\mu}}\mathcal{L}_\tau(\bm{0};\bm{0})_{kj} = \frac{1}{K^2\tau^4}I_d,
    \qquad k\neq j.
\end{align}
Therefore, for any direction $\bm{h}=(h_\ell)_{\ell=1}^K$,
\begin{align}
    \left.\frac{d^2}{dt^2}\mathcal{L}_\tau(t\bm{h};\bm{0})\right|_{t=0} &= \langle \bm{h},  \nabla^2_{\bm{\mu}}\mathcal{L}_\tau(\bm{0};\bm{0})\bm{h}\rangle
    \notag\\
    &=
    \frac{1}{K\tau^2}\sum_{\ell=1}^K\|h_\ell\|^2 - \frac{K-1}{K^2\tau^4}\sum_{\ell=1}^K\|h_\ell\|^2 + \frac{1}{K^2\tau^4}\sum_{\ell\neq j}\langle h_\ell,h_j\rangle.
\end{align}
Using
\begin{align}
    \sum_{\ell\neq j}\langle h_\ell,h_j\rangle = \left\|\sum_{\ell=1}^K h_\ell\right\|^2 - \sum_{\ell=1}^K\|h_\ell\|^2 = -\sum_{\ell=1}^K\|h_\ell\|^2,
\end{align}
we obtain
\begin{align}
    \left.\frac{d^2}{dt^2}\mathcal{L}_\tau(t\bm{h};\bm{0})\right|_{t=0} = -\frac{1-\tau^2}{K\tau^4}\sum_{\ell=1}^K \|h_\ell\|^2.
\end{align}
Since $\tau\in(0,1)$ and $\bm{h}\neq\bm{0}$, this quantity is strictly negative. Hence, $\bm{0}$ is not a local minimizer of $\mathcal{L}_\tau(\,\cdot\,;\bm{0})$.
\end{proof}

We next show that this non-minimality persists under weak signal: for sufficiently small signal-strength $\beta$, minimizers of $\mathcal{L}_\tau(\cdot;\beta\bm{\mu}^\star)$ on a fixed compact set cannot approach the collapsed configuration.

\begin{lem}
\label{lem:persistence_uniform_convergence}
Fix $\rho\in(0,1)$, and let $V\subset(\mathbb{R}^d)^K$ be compact with $\bm{0}\in V$.
Then, there exist $\eta_\rho>0$ and $\beta_0>0$ such that every minimizer $\widehat{\bm{\mu}}(\beta)\in \argmin_{\bm{\mu}\in V}\mathcal{L}_\tau(\bm{\mu};\beta\bm{\mu}^\star)$ satisfies
\begin{align}
    \|\widehat{\bm{\mu}}(\beta)\|_F \ge \eta_\rho,
    \qquad \forall\,\beta\in[0,\beta_0].
\end{align}
\end{lem}

\begin{proof}[Proof of Lemma~\ref{lem:persistence_uniform_convergence}]
We first show that
\begin{align}
    \sup_{\bm{\mu}\in V} \big| \mathcal{L}_\tau(\bm{\mu};\beta\bm{\mu}^\star) - \mathcal{L}_\tau(\bm{\mu};\bm{0}) \big| \xrightarrow[\beta\to 0]{}0.
    \label{eq:uniform_conv_beta_final}
\end{align}
Indeed, since $p_{\beta\bm{\mu}^\star,1}(x)\to \varphi_d(x;0,I_d)$ pointwise as $\beta\to0$, and since $V$ is compact, there exists $C_V>0$ such that
\begin{align}
    |\log p_{\bm{\mu},\rho}(x)| \le C_V(1+\|x\|^2)
    \qquad
    \text{for all } \bm{\mu}\in V.
\end{align}
Therefore, the family $\{\log p_{\bm{\mu},\rho}(\cdot):\bm{\mu}\in V\}$ is uniformly dominated by an integrable envelope, and dominated convergence yields \eqref{eq:uniform_conv_beta_final}.

By Lemma~\ref{lem:negative_curvature_zero_signal}, $\bm{0}$ is not a minimizer of $\mathcal{L}_\tau(\cdot;\bm{0})$ over $V$, and by the continuity of $\mathcal{L}_\tau(\cdot;\bm{0})$ on the compact set $V$, there exists a neighborhood of $\bm{0}$ on which $\mathcal{L}_\tau(\cdot;\bm{0})$ is bounded strictly above its minimum over $V$.
By the uniform convergence in \eqref{eq:uniform_conv_beta_final}, the same strict separation persists for $\mathcal{L}_\tau(\cdot;\beta\bm{\mu}^\star)$ for all sufficiently small $\beta$.
Hence, every minimizer of $\mathcal{L}_\tau(\cdot;\beta\bm{\mu}^\star)$ over $V$ must stay a positive distance away from $\bm{0}$, proving the claim.
\end{proof}

\subsection{Proof of Proposition~\ref{prop:hard_assignment_limit}} \label{sec:proof_hard_assignment_limit}

We prove the proposition in three steps.

\paragraph{Step 1: Decomposition into hard-assignment risk and remainder.}
By the definition of $\mathcal{L}_\tau$~\eqref{eq:def_L_tau} and of $p_{\bm{\mu}, \tau}$ in~\eqref{eq:gmm_isotropic_density},
\begin{align}
    \log p_{\bm{\mu},\tau}(y) = -\frac{d}{2}\log(2\pi\tau^2)-\log K + \log \sum_{\ell\in[K]} \exp\!\left(-\frac{\|y-\mu_\ell\|^2}{2\tau^2}\right),
\end{align}
and
\begin{align}
    \label{eqn:app_A2}
    \mathcal{L}_\tau(\bm{\mu}; \bm{\mu^\star})
    &=
    \frac{d}{2}\log(2\pi\tau^{2}) +\log K -\mathbb{E}\!\left[\log\!\left(\sum_{\ell\in[K]} \exp\!\left(-\frac{\|Y-\mu_{\ell}\|^{2}}{2\tau^{2}} \right) \right) \right].
\end{align}
Factoring out the minimal squared distance from $Y$ to the candidate centers inside the log-sum-exp term results
\begin{align}
    \|Y-\mu_\ell\|^2 = \min_{j\in[K]}\|Y-\mu_j\|^2 + \Big(\|Y-\mu_\ell\|^2-\min_{j\in[K]}\|Y-\mu_j\|^2\Big),
\end{align}
we obtain
\begin{align}
    \sum_{\ell\in[K]}
    & \exp\!\left(-\frac{\|Y-\mu_\ell\|^2}{2\tau^2}\right)
    \\&=
    \exp\!\left(-\frac{\min_{j\in[K]}\|Y-\mu_j\|^2}{2\tau^2}\right) \sum_{\ell\in[K]} \exp\!\left(-\frac{\|Y-\mu_\ell\|^2-\min_{j\in[K]}\|Y-\mu_j\|^2}{2\tau^2} \right).
\end{align}
Substituting this identity into~\eqref{eqn:app_A2} gives
\begin{align}
    \mathcal{L}_\tau(\bm{\mu}; \bm{\mu^\star}) = \frac{d}{2}\log(2\pi\tau^2) +\log K +\frac{1}{2\tau^2}\Phi(\bm{\mu}) +r_\tau(\bm{\mu}),
\end{align}
where $\Phi(\bm{\mu})=\mathbb{E}[\min_{j\in[K]}\|Y-\mu_j\|^2]$ is the population $k$-means risk~\eqref{eq:kmeans_population} and $r_\tau(\bm{\mu})$ is defined by~\eqref{eq:rtau_def}. This proves~\eqref{eq:Ltau_expansion_pointwise}. Moreover, since for every $\ell\in[K]$, $\|Y-\mu_\ell\|^2-\min_{j\in[K]}\|Y-\mu_j\|^2\ge 0$, and for at least one index this quantity is equal to zero, we have
\begin{align}
    \label{eq:r_tau_bound}
    -\log K \le r_\tau(\bm{\mu}) \le 0.
\end{align}

\paragraph{Step 2: Point-wise convergence of $r_\tau$.}
Since the means in $\bm{\mu}$ are distinct, the nearest center to $Y$ is unique almost surely (see Remark~\ref{remark:arbitrary_tie_breaking}). Therefore, for almost every $Y$, exactly one term in
\begin{align}
    \sum_{\ell\in[K]} \exp\!\left(-\frac{\|Y-\mu_\ell\|^2-\min_{j\in[K]}\|Y-\mu_j\|^2}{2\tau^2} \right)
\end{align}
is equal to $1$, while all the others converge to $0$ as $\tau\to 0^+$. Hence,
\begin{align}
    \sum_{\ell\in[K]} \exp\!\left(-\frac{\|Y-\mu_\ell\|^2-\min_{j\in[K]}\|Y-\mu_j\|^2}{2\tau^2} \right) \xrightarrow[]{\mathrm{a.s.}} 1.
\end{align}
Thus, the integrand of the expectation in the definition of $r_\tau(\bm{\mu})$ converges almost surely to $0$. Moreover, this integrand is bounded in absolute value by $\log K$. The conclusion follows from dominated convergence.

\paragraph{Step 3: Convergence of minimizers.}
Define the rescaled objective
$F_\tau(\bm{\mu})\triangleq 2\tau^2\,\mathcal{L}_\tau(\bm{\mu}; \bm{\mu^\star}) - d\tau^2\log(2\pi\tau^2) - 2\tau^2\log K$, which has the same minimizers as $\mathcal{L}_\tau$ over $\mathcal{U}$. By~\eqref{eq:Ltau_expansion_pointwise}, $F_\tau(\bm{\mu})=\Phi(\bm{\mu})+2\tau^2 r_\tau(\bm{\mu})$. Hence, by~\eqref{eq:r_tau_bound},
\begin{align}
    \sup_{\bm{\mu}\in\mathcal{U}}\big|F_\tau(\bm{\mu})-\Phi(\bm{\mu})\big| \le 2\tau^2\log K.    \label{eq:proof_uniform_conv_Ftau}
\end{align}
Thus as $\tau \to 0$, $F_\tau\to\Phi$ uniformly on $\mathcal{U}$.

Let $\bm{\mu}^{(\tau)}\in\argmin_{\bm{\mu}\in\mathcal{U}}F_\tau(\bm{\mu})$, and let $m\triangleq \min_{\bm{\mu}\in\mathcal{U}}\Phi(\bm{\mu})=\Phi(\bm{\mu}_{\mathrm{HA}}^\star)$. Fix $\varepsilon>0$. Since $\bm{\mu}_{\mathrm{HA}}^\star$ is the unique minimizer of $\Phi$ over $\mathcal{U}$ up to permutation, continuity of $\Phi$ and compactness of $\mathcal{U}$ imply that there exists $\delta_\varepsilon>0$ such that $d_{\mathrm{perm}}(\bm{\mu},\bm{\mu}_{\mathrm{HA}}^\star)\ge\varepsilon$ implies $\Phi(\bm{\mu})\ge m+\delta_\varepsilon$.
If $d_{\mathrm{perm}}(\bm{\mu}^{(\tau)},\bm{\mu}_{\mathrm{HA}}^\star)\ge\varepsilon$, then by~\eqref{eq:proof_uniform_conv_Ftau},
\begin{align}
    F_\tau(\bm{\mu}^{(\tau)}) \ge m+\delta_\varepsilon-2\tau^2\log K,
    \qquad
    F_\tau(\bm{\mu}_{\mathrm{HA}}^\star) \le m+2\tau^2\log K.
\end{align}
Since $F_\tau(\bm{\mu}^{(\tau)})\le F_\tau(\bm{\mu}_{\mathrm{HA}}^\star)$, it follows that $\delta_\varepsilon\le 4\tau^2\log K$, which is impossible for all sufficiently small $\tau>0$. Therefore $d_{\mathrm{perm}}(\bm{\mu}^{(\tau)},\bm{\mu}_{\mathrm{HA}}^\star)<\varepsilon$ for all sufficiently small $\tau$. Since $\varepsilon>0$ was arbitrary, this proves~\eqref{eq:mu_tau_to_mu_HA}.

\section{Mean estimation under variance misspecification: proofs}

\subsection{Proof of Proposition~\ref{prop:collapse_generalK}}
\label{sec:proof_collapse_generalK}

By translation invariance, we may work in centered coordinates and assume without loss of generality that $\bar{\mu}^\star = \frac{1}{K}\sum_{\ell\in[K]}\mu_\ell^\star = 0$.
In these coordinates, the collapsed configuration is $\bm{\mu}_{\mathrm{coll}}=\bm{0}$, $\mathbb{E}[Y]=0$, and $\mathbb{E}[YY^\top]=\sigma^2 I_d+\Sigma_\mu$.
Thus $\bm{\mu}_{\mathrm{coll}}$ is a stationary point by symmetry and the fact that the population mean is zero.

To determine if it is a local minima, we analyze the Hessian at $\bm{\mu}=\bm{0}$. By Lemma~\ref{lem:hessian_blocks}, its $d\times d$ blocks are
\begin{align}
    H_{kk} &= \frac{I_d}{K\tau^2} -\frac{(K-1)\,\mathbb{E}[YY^\top]}{K^2\tau^4}, \\
    H_{kj}
    &= \frac{\mathbb{E}[YY^\top]}{K^2\tau^4},
    \qquad k\neq j.
\end{align}

Let $\bm{h}=(h_1,\dots,h_K)\in(\mathbb{R}^d)^K$. Then,
\begin{align}
    \bm{h}^\top H\bm{h}
    &=
    \sum_{\ell=1}^K h_\ell^\top H_{\ell\ell} h_\ell
    +\sum_{\ell\neq j} h_\ell^\top H_{\ell j} h_j
    \notag\\
    &=
    \frac{1}{K\tau^2}\sum_{\ell=1}^K \|h_\ell\|^2
    -\frac{K-1}{K^2\tau^4}\sum_{\ell=1}^K h_\ell^\top \mathbb{E}[YY^\top] h_\ell
    +\frac{1}{K^2\tau^4}\sum_{\ell\neq j} h_\ell^\top \mathbb{E}[YY^\top] h_j.
\end{align}
Using
\begin{align}
    \sum_{\ell\neq j} h_\ell^\top \mathbb{E}[YY^\top] h_j
    =
    \left(\sum_{\ell=1}^K h_\ell\right)^\top
    \mathbb{E}[YY^\top]
    \left(\sum_{\ell=1}^K h_\ell\right)
    -
    \sum_{\ell=1}^K h_\ell^\top \mathbb{E}[YY^\top] h_\ell,
\end{align}
we obtain
\begin{align}
    \bm{h}^\top H\bm{h} = \frac{1}{K\tau^2}\sum_{\ell=1}^K \|h_\ell\|^2 -\frac{1}{K\tau^4}\sum_{\ell=1}^K h_\ell^\top \mathbb{E}[YY^\top] h_\ell +\frac{1}{K^2\tau^4} \left(\sum_{\ell=1}^K h_\ell\right)^\top \mathbb{E}[YY^\top] \left(\sum_{\ell=1}^K h_\ell\right).
\end{align}

Now decompose $(\mathbb{R}^d)^K$ as the orthogonal direct sum of the common-shift subspace
\begin{align}
    \mathcal{S} \triangleq \{(u,\dots,u):u\in\mathbb{R}^d\},
\end{align}
and the zero-sum subspace
\begin{align}
    \mathcal{S}_0 \triangleq \left\{ (h_1,\dots,h_K)\in(\mathbb{R}^d)^K: \sum_{\ell=1}^K h_\ell=0 \right\}.
\end{align}
That is, $(\mathbb{R}^d)^K = \mathcal{S} \oplus \mathcal{S}_0$, and $\mathcal{S}^\perp=\mathcal{S}_0$.
We will show below that the Hessian is strictly positive on $\mathcal{S}$ and positive semidefinite on $\mathcal{S}_0$ exactly when $\tau^2 I_d-\mathbb{E}[YY^\top]\succeq0$, the origin configuration is a local minimizer if and only if this latter condition holds.

\paragraph{Step 1: Common-shift direction.}
For a common-shift direction $\bm{h}=(u,\dots,u)$, a direct computation gives
\begin{align}
    \bm{h}^\top H\bm{h} = \frac{1}{\tau^2}\|u\|^2 >0.
\end{align}
Thus the Hessian is strictly positive on $\mathcal{S}$.

\paragraph{Step 2: Zero-sum direction.}
For a zero-sum direction $\bm{h}\in\mathcal{S}_0$, the last term vanishes, and therefore
\begin{align}
    \bm{h}^\top H\bm{h} = \frac{1}{K\tau^2}\sum_{\ell=1}^K \|h_\ell\|^2 -\frac{1}{K\tau^4}\sum_{\ell=1}^K h_\ell^\top \mathbb{E}[YY^\top] h_\ell.
\end{align}
Hence the Hessian is positive semidefinite on $\mathcal{S}_0$ if and only if $\tau^2 I_d - \mathbb{E}[YY^\top] \succeq 0$, or equivalently, $\tau^2 \ge \lambda_{\max}\!\big(\mathbb{E}[YY^\top]\big)$.
Since $\mathbb{E}[YY^\top]=\sigma^2 I_d+\Sigma_\mu$, this condition becomes
\begin{align}
    \tau^2 \ge \sigma^2+\lambda_{\max}(\Sigma_\mu),
\end{align}
or, after dividing by $\sigma^2$,
\begin{align}
    \rho^2 = \frac{\tau^2}{\sigma^2} \ge 1+\frac{\lambda_{\max}(\Sigma_\mu)}{\sigma^2}.
\end{align}
This proves part~(1).

\paragraph{Step 3: proof of~\eqref{eqn:lambda_max_bounds}.}
For part~(2) and~\eqref{eqn:lambda_max_bounds}, since $\Sigma_\mu\succeq 0$,
\begin{align}
    \frac{\tr(\Sigma_\mu)}{d} \le \lambda_{\max}(\Sigma_\mu) \le \tr(\Sigma_\mu).
\end{align}
Using $\tr(\Sigma_\mu)=\sigma^2\,\mathrm{SNR}$, we obtain
\begin{align}
    1+\frac{\mathrm{SNR}}{d} \le 1+\frac{\lambda_{\max}(\Sigma_\mu)}{\sigma^2} \le 1+\mathrm{SNR}.
\end{align}
This completes the proof.

\subsection{Proof of Corollary~\ref{cor:collapse_K2}}
\label{sec:proof_collapse_K2}
For $K=2$ with $(\mu_0^\star,\mu_1^\star)=(\mu,-\mu)$, we have $\bar{\mu}^\star=0$ and
\begin{align}
    \Sigma_\mu = \frac{1}{2}\big(\mu\mu^\top+(-\mu)(-\mu)^\top\big) =  \mu\mu^\top.
\end{align}
Hence $\lambda_{\max}(\Sigma_\mu)=\|\mu\|^2$. Since in this symmetric two-component model $\mathrm{SNR}=\frac{\|\mu\|^2}{\sigma^2}$, the claim follows immediately from Proposition~\ref{prop:collapse_generalK}.

\subsection{Proof of Corollary~\ref{cor:collapse_simplex}}
\label{sec:proof_collapse_simplex}

Since $\mu_\ell^\star=\beta v_\ell$ and $\sum_{\ell}v_\ell=0$, we have $\bar{\mu}^\star=0$, and therefore
\begin{align}
    \Sigma_\mu =\frac{\beta^2}{K}\sum_{\ell\in[K]} v_\ell v_\ell^\top.
\end{align}
For a regular simplex, $\frac{1}{K}\sum_{\ell\in[K]} v_\ell v_\ell^\top$ is the orthogonal projector onto $\mathrm{span}\{v_\ell\}_{\ell\in[K]}$, scaled by $1/(K-1)$. Hence $\lambda_{\max}(\Sigma_\mu)= \beta^2/(K-1)$.
Since for this model $\mathrm{SNR}=\beta^2/\sigma^2$, the conclusion follows from Proposition~\ref{prop:collapse_generalK}.

\subsection{Proof of Theorem~\ref{thm:lowSNR_under_smoothed_fixed_rho}}
\label{sec:proof_lowSNR_under_smoothed_fixed_rho}

It is convenient to work in the normalized model with unit data variance, $\sigma^2=1$, and encode the signal level by a scalar parameter $\beta\ge 0$; that is, the data are drawn from the unit-variance mixture $p_{\beta\bm{\mu}^\star,1}$, while the fitted model uses variance $\rho^2$, where $\rho\in(0,1)$ is fixed. Accordingly,
\begin{align}
    \mathcal{L}_\rho(\bm{\mu};\beta\bm{\mu}^\star) \triangleq -\mathbb{E}_{Y\sim p_{\beta\bm{\mu}^\star,1}} \big[\log p_{\bm{\mu},\rho}(Y)\big].
\end{align}
In this parameterization,
\begin{align}
    \mathrm{SNR} = \frac{1}{K}\sum_{\ell\in[K]} \|\beta\mu_\ell^\star-\beta\bar{\mu}^\star\|^2 = \beta^2\cdot \frac{1}{K}\sum_{\ell\in[K]} \|\mu_\ell^\star-\bar{\mu}^\star\|^2,
\end{align}
so the low-SNR regime is exactly $\beta\to 0$. At $\beta=0$,
\begin{align}\label{eq:Lrho_zero_signal}
    \mathcal{L}_\rho(\bm{\mu};0) = -\mathbb{E}_{X\sim\mathcal{N}(0,I_d)}  \big[\log p_{\bm{\mu},\rho}(X)\big].
\end{align}
By the scaling reduction, it is enough to prove that in the normalized model,
\begin{align}
    d_{\mathrm{perm}}^2(\widehat{\bm{\mu}}(\beta),\beta\bm{\mu}^\star)=\Omega(1)
    \qquad\text{as }\beta\to0,
    \label{eq:normalized_omega_one}
\end{align}
where $\widehat{\bm{\mu}}(\beta) \in \argmin_{\bm{\mu}\in(\mathbb{R}^d)^K} \mathcal{L}_\rho(\bm{\mu};\beta\bm{\mu}^\star)$.

By Lemma~\ref{lem:persistence_uniform_convergence}, there exist constants $\eta_\rho>0$ and $\beta_0>0$ such that every minimizer $\widehat{\bm{\mu}}(\beta)$ satisfies $\|\widehat{\bm{\mu}}(\beta)\|_F\ge \eta_\rho$, for all $\beta\in[0,\beta_0]$.
Since $d_{\mathrm{perm}}(\bm{\mu},0)=\|\bm{\mu}\|_F$ and $\beta\bm{\mu}^\star\to0$, the triangle inequality gives
\begin{align}
    d_{\mathrm{perm}}(\widehat{\bm{\mu}}(\beta),\beta\bm{\mu}^\star)
    &\ge
    d_{\mathrm{perm}}(\widehat{\bm{\mu}}(\beta),0) - d_{\mathrm{perm}}(\beta\bm{\mu}^\star,0)
    \notag\\
    &=
    \|\widehat{\bm{\mu}}(\beta)\|_F-\beta\|\bm{\mu}^\star\|_F.
\end{align}
Hence, for all sufficiently small $\beta$, we have $d_{\mathrm{perm}}(\widehat{\bm{\mu}}(\beta),\beta\bm{\mu}^\star) \ge \eta_\rho/2$, and therefore
\begin{align}
    d_{\mathrm{perm}}^2(\widehat{\bm{\mu}}(\beta),\beta\bm{\mu}^\star) \ge \frac{\eta_\rho^2}{4}.
    \label{eq:normalized_lower_bound}
\end{align}
This proves~\eqref{eq:normalized_omega_one}.

Finally, return to the original variables. With $\beta=\sigma^{-1}$ and $\tau=\rho\sigma$, distances scale by $\sigma$, 
\begin{align}
    d_{\mathrm{perm}}^2(\bm{\mu}_{\tau}^\star,\bm{\mu}^\star) = \sigma^2\, d_{\mathrm{perm}}^2(\widehat{\bm{\mu}}(\beta),\beta\bm{\mu}^\star).
\end{align}
Combining this identity with~\eqref{eq:normalized_lower_bound}, we conclude that there exist constants $C_\rho>0$ and $\sigma_0<\infty$, depending on $K$, $d$, $\bm{\mu}^\star$, and $\rho$, such that
\begin{align}
    d_{\mathrm{perm}}^{2}\!\left(\bm{\mu}_{\tau}^\star,\bm{\mu}^\star\right) \ge C_\rho\,\sigma^2,
    \qquad \forall\,\sigma\ge \sigma_0.
\end{align}
This proves the theorem.

\subsection{Proof of Proposition~\ref{prop:highSNR_under_smoothed_upper_uniform}}
\label{sec:proof-prop-highSNR_under_smoothed_upper_uniform}

We divide the proof into three steps.

\paragraph{Step 1: uniform localization and preliminary identities.}
For each fixed $\tau\in(0,\sigma]$, the data-generating law $p_{\bm{\mu}^\star,\sigma}$ converges  as $\sigma\to 0$, to the discrete measure
\begin{align}
    \frac{1}{K}\sum_{\ell\in[K]}\delta_{\mu_\ell^\star}.
\end{align}
Accordingly, the population objective $\mathcal{L}_{\tau}(\bm{\mu},\bm{\mu}^\star)$ converges pointwise, up to additive constants independent of $\bm{\mu}$, to the corresponding discrete objective obtained by evaluating the fitted model on the atoms $\mu_\ell^\star$. The unique minimizers of that limiting objective are exactly the permutations of $\bm{\mu}^\star$. Therefore, by compactness of $\mathcal{U}$ and Assumption~\ref{ass:unique_minimizers_up_to_perm},
\begin{align}
    \sup_{0<\tau\le \sigma} d_{\mathrm{perm}}\!\bigl(\bm{\mu}_{\tau}^\star,\bm{\mu}^\star\bigr)  \xrightarrow[\sigma\to0]{}0,    \label{eq:uniform_localization_sigma_to_zero_tau}
\end{align}
proving~\eqref{eq:high_SNR_conv}.
Fix now $\varepsilon\in(0,\Delta_{\min}/4)$. By \eqref{eq:uniform_localization_sigma_to_zero_tau}, there exists $\sigma_\varepsilon>0$ such that for all $\sigma\le \sigma_\varepsilon$ and all $\tau\in(0,\sigma]$, one can choose a permutation $\pi=\pi(\tau,\sigma)$ satisfying
\begin{align}
    \max_{\ell\in[K]} \bigl\| \mu_{\tau,\pi(\ell)}^\star-\mu_\ell^\star \bigr\| \le \varepsilon.    \label{eq:localized_matching_under_smoothed_tau}
\end{align}

For the remainder of the proof, fix $\tau\in(0,\sigma]$, $\sigma\le \sigma_\varepsilon$, and let $\pi$ be as in \eqref{eq:localized_matching_under_smoothed_tau}. Define the posterior responsibilities
\begin{align}
    \gamma_r(y) \triangleq \frac{\exp\!\left(-\|y-\mu_{\tau,r}^\star\|^2/(2\tau^2)\right)}  {\sum_{s\in[K]}\exp\!\left(-\|y-\mu_{\tau,s}^\star\|^2/(2\tau^2)\right)},
    \qquad r\in[K],
\end{align}
and write $m_\ell \triangleq \mu_{\tau,\pi(\ell)}^\star$ and $\gamma_\ell(Y) \triangleq \gamma_{\pi(\ell)}(Y)$.
Since $\bm{\mu}_\tau^\star$ minimizes $\mathcal{L}_{\tau}(\bm{\mu},\bm{\mu}^\star)$, the first-order stationarity condition with respect to $\mu_{\tau,\pi(\ell)}^\star$ gives
\begin{align}
    0 = \nabla_{\mu_{\tau,\pi(\ell)}^\star}\mathcal{L}_{\tau}(\bm{\mu}_{\tau}^\star,\bm{\mu}^\star) = -\frac{1}{\tau^2}\, \mathbb{E}\!\left[\gamma_\ell(Y)\bigl(Y-m_\ell\bigr)\right].
\end{align}
Equivalently,
\begin{align}
    m_\ell = \frac{\mathbb{E}[\gamma_\ell(Y)Y]}{\mathbb{E}[\gamma_\ell(Y)]},    \label{eq:soft_mean_identity_step3_tau}
\end{align}
and hence
\begin{align}
    m_\ell-\mu_\ell^\star = \frac{\mathbb{E}[\gamma_\ell(Y)(Y-\mu_\ell^\star)]}{\mathbb{E}[\gamma_\ell(Y)]}.    \label{eq:mean_error_basic_step3_tau}
\end{align}
Finally, define
\begin{align}
    q_\ell(\tau,\sigma) \triangleq \mathbb{E}\!\left[1-\gamma_{\pi(\ell)}(Y)\mid L=\ell\right] + \sum_{j\neq \ell}\mathbb{E}\!\left[\gamma_{\pi(\ell)}(Y)\mid L=j\right].
\label{eq:def_qell_tau}
\end{align}
Thus, $q_\ell(\tau,\sigma)$ measures the total soft-assignment mass that is misallocated relative to the matched pair $(\ell,\pi(\ell))$.

\paragraph{Step 2: bound on the misallocated soft-assignment mass $q_\ell(\tau,\sigma)$.}
For $\ell\neq j$, write
\begin{align}
    u_{\ell j}  \triangleq \frac{\mu_j^\star-\mu_\ell^\star}{\|\mu_j^\star-\mu_\ell^\star\|},
    \qquad
    \Delta_{\ell j} \triangleq \|\mu_j^\star-\mu_\ell^\star\|.
\end{align}
By \eqref{eq:localized_matching_under_smoothed_tau}, the matched estimated mean $m_\ell=\mu_{\tau,\pi(\ell)}^\star$ lies within $\varepsilon$ of $\mu_\ell^\star$, while every competing matched mean $m_j=\mu_{\tau,\pi(j)}^\star$, $j\neq \ell$, lies within $\varepsilon$ of $\mu_j^\star$. Therefore, for a sample $Y=\mu_\ell^\star+\sigma Z$ with $Z\sim\mathcal{N}(0,I_d)$, substantial posterior mass can be transferred from $\pi(\ell)$ to a mismatched component $\pi(j)$ only if $Y$ crosses the corresponding pairwise decision boundary. Likewise, for a sample generated from component $j\neq \ell$, substantial posterior mass can be assigned to $\pi(\ell)$ only if the same boundary is crossed. In either case, this implies
\begin{align}
    \langle Z,u_{\ell j}\rangle \ge \frac{\Delta_{\ell j}-2\varepsilon}{2\sigma}.
\end{align}
Consequently,
\begin{align}
    q_\ell(\tau,\sigma)
    &\le
    \sum_{j\neq \ell} \mathbb{P}\!\left( \langle Z,u_{\ell j}\rangle \ge \frac{\Delta_{\ell j}-2\varepsilon}{2\sigma}
    \right).
\label{eq:qell_union_bound_tau}
\end{align}
Since $\langle Z,u_{\ell j}\rangle\sim \mathcal{N}(0,1)$, the standard Gaussian tail bound gives
\begin{align}
    \mathbb{P}\!\left(\langle Z,u_{\ell j}\rangle \ge  \frac{\Delta_{\ell j}-2\varepsilon}{2\sigma}  \right)
    &\le
    \exp\!\left( -\frac{(\Delta_{\ell j}-2\varepsilon)^2}{8\sigma^2} \right).
\end{align}
Using $\Delta_{\ell j}\ge \Delta_{\min}$, we obtain
\begin{align}
    q_\ell(\tau,\sigma)
    &\le (K-1)\exp\!\left( -\frac{(\Delta_{\min}-2\varepsilon)^2}{8\sigma^2}
\right).
\label{eq:qell_bound_eps_tau}
\end{align}
Since $\varepsilon>0$ is arbitrary, for every $\eta>0$ we may choose it small enough so that
\begin{align}
    \frac{(\Delta_{\min}-2\varepsilon)^2}{8}  \ge  \Bigl(\frac{1}{8}-\eta\Bigr)\Delta_{\min}^2.
    \label{eq:eps_choice_eta_tau}
\end{align}
Thus, we obtain
\begin{align}
    q_\ell(\tau,\sigma) \le (K-1) \exp\!\left( -\Bigl(\frac{1}{8}-\eta\Bigr)\frac{\Delta_{\min}^2}{\sigma^2} \right)
    \label{eq:qell_bound_eta_tau}
\end{align}
for all $\ell\in[K]$, all $\tau\in(0,\sigma]$, and all $\sigma\le \sigma_\varepsilon$.

\paragraph{Step 3: converting misallocated soft-assignment mass into a centers bound.}
Next, we bound the numerator and denominator of~\eqref{eq:mean_error_basic_step3_tau}. We start with the numerator. Define
\begin{align}
    A_\ell(Y,L) \triangleq
    \begin{cases}
        (\gamma_\ell(Y)-1)(Y-\mu_\ell^\star), & L=\ell,\\[1ex]
        \gamma_\ell(Y)(Y-\mu_\ell^\star), & L\neq \ell.
    \end{cases}
\end{align}
Using $\mathbb{E}[Y-\mu_\ell^\star\mid L=\ell]=0$, we may rewrite
\begin{align}
    \mathbb{E}[\gamma_\ell(Y)(Y-\mu_\ell^\star)]  = \mathbb{E}[A_\ell(Y,L)].
    \label{eq:numerator_as_A}
\end{align}
Therefore, by Jensen's inequality,
\begin{align}
    \left\|  \mathbb{E}[\gamma_\ell(Y)(Y-\mu_\ell^\star)] \right\|^2  \le \mathbb{E}\!\left[\|A_\ell(Y,L)\|^2\right].
\label{eq:jensen_A}
\end{align}
Conditioning on $L$, we obtain
\begin{align}
    \mathbb{E}\!\left[\|A_\ell(Y,L)\|^2\right]
    &=
    \frac{1}{K}\,  \mathbb{E}\!\left[(1-\gamma_\ell(Y))^2\|Y-\mu_\ell^\star\|^2\mid L=\ell\right]
    \nonumber\\
    &\quad+
    \frac{1}{K}\sum_{j\neq \ell} \mathbb{E}\!\left[\gamma_\ell(Y)^2\|Y-\mu_\ell^\star\|^2\mid L=j\right].
\label{eq:A_second_moment}
\end{align}
Since $0\le \gamma_\ell(Y)\le 1$, we have $(1-\gamma_\ell(Y))^2\le 1-\gamma_\ell(Y)$, as well as $\gamma_\ell(Y)^2\le \gamma_\ell(Y)$.
Hence,
\begin{align}
    \left\| \mathbb{E}[\gamma_\ell(Y)(Y-\mu_\ell^\star)] \right\|^2
    &\le
    \frac{1}{K}\, \mathbb{E}\!\left[(1-\gamma_\ell(Y))\|Y-\mu_\ell^\star\|^2\mid L=\ell\right]
    \nonumber\\
    &\quad+
    \frac{1}{K}\sum_{j\neq \ell} \mathbb{E}\!\left[\gamma_\ell(Y)\|Y-\mu_\ell^\star\|^2\mid L=j\right].    \label{eq:numerator_squared_bound_step3_tau}
\end{align}

If $L=\ell$, then $Y=\mu_\ell^\star+\sigma Z$, so
\begin{align}
    \mathbb{E}\!\left[\|Y-\mu_\ell^\star\|^2\mid L=\ell\right] = \sigma^2 d.  \label{eq:same_component_second_moment_step3_tau}
\end{align}
If $L=j\neq \ell$, then
\begin{align}
    Y-\mu_\ell^\star  = (\mu_j^\star-\mu_\ell^\star)+\sigma Z,
\end{align}
and therefore
\begin{align}
    \mathbb{E}\!\left[\|Y-\mu_\ell^\star\|^2\mid L=j\right] =  \|\mu_j^\star-\mu_\ell^\star\|^2+\sigma^2 d \le   \Delta_{\max}^2+\sigma^2 d.  \label{eq:cross_component_second_moment_step3_tau}
\end{align}
Using \eqref{eq:def_qell_tau}, \eqref{eq:same_component_second_moment_step3_tau}, and \eqref{eq:cross_component_second_moment_step3_tau} in \eqref{eq:numerator_squared_bound_step3_tau}, and noting that $\sigma^2 d\le \Delta_{\max}^2+\sigma^2 d$, we obtain
\begin{align}
    \left\| \mathbb{E}[\gamma_\ell(Y)(Y-\mu_\ell^\star)] \right\|^2
    &\le \frac{1}{K}\,    q_\ell(\tau,\sigma)\bigl(\Delta_{\max}^2+\sigma^2 d\bigr).    \label{eq:numerator_bound_step3_tau}
\end{align}

We next lower-bound the denominator in \eqref{eq:mean_error_basic_step3_tau}. Since
\begin{align}
    \mathbb{E}[\gamma_\ell(Y)]
    &=
    \frac{1}{K}\mathbb{E}[\gamma_\ell(Y)\mid L=\ell] + \frac{1}{K}\sum_{j\neq \ell}\mathbb{E}[\gamma_\ell(Y)\mid L=j]
    \nonumber\\
    &\ge
    \frac{1}{K}\mathbb{E}[\gamma_\ell(Y)\mid L=\ell] = \frac{1}{K}\Bigl(1-\mathbb{E}[1-\gamma_\ell(Y)\mid L=\ell]\Bigr)
    \nonumber\\
    &\ge
    \frac{1}{K}\bigl(1-q_\ell(\tau,\sigma)\bigr),   \label{eq:denominator_lower_step3_tau}
\end{align}
and $q_\ell(\tau,\sigma)\to0$ exponentially fast by \eqref{eq:qell_bound_eta_tau}. Hence there exists $\sigma_1>0$ such that for all $\sigma\le \sigma_1$,
\begin{align}
    \mathbb{E}[\gamma_\ell(Y)]\ge \frac{1}{2K}.    \label{eq:denominator_uniform_lower_step3_tau}
\end{align}
Combining \eqref{eq:mean_error_basic_step3_tau}, \eqref{eq:numerator_bound_step3_tau}, and \eqref{eq:denominator_uniform_lower_step3_tau}, we conclude that for all sufficiently small $\sigma$,
\begin{align}
    \|m_\ell-\mu_\ell^\star\|^2 \le 4K\bigl(\Delta_{\max}^2+\sigma^2 d\bigr)\,q_\ell(\tau,\sigma).\label{eq:mean_from_leakage_soft_tau}
\end{align}

Finally, combining \eqref{eq:mean_from_leakage_soft_tau} with
\eqref{eq:qell_bound_eta_tau} and summing over $\ell\in[K]$, we obtain
\begin{align}
    d_{\mathrm{perm}}^2\!\bigl(\bm{\mu}_{\tau}^\star,\bm{\mu}^\star\bigr) &\le  4K^2(K-1)\bigl(\Delta_{\max}^2+\sigma^2 d\bigr) \exp\!\left( -\Bigl(\frac{1}{8}-\eta\Bigr)\frac{\Delta_{\min}^2}{\sigma^2} \right),
\end{align}
uniformly over $0<\tau\le \sigma$, for all sufficiently small $\sigma$. This proves
\eqref{eq:highSNR_under_smoothed}.

\subsection{Proof of Corollary~\ref{prop:highSNR_mean_upper}}
\label{subsec:proof_highSNR_HA}
By Proposition~\ref{prop:hard_assignment_limit}, for every fixed $\sigma>0$,
\begin{align}
    d_{\mathrm{perm}}\!\bigl(\bm{\mu}_{\tau}^\star,\bm{\mu}_{\mathrm{HA}}^\star(\sigma)\bigr) \xrightarrow[\tau\to0]{}0.    \label{eq:tau_to_zero_hard_assignment_limit}
\end{align}
Fix $\sigma>0$. Since the bound in Proposition~\ref{prop:highSNR_under_smoothed_upper_uniform} holds uniformly over all $0<\tau\le \sigma$, we may let $\tau\to0$ in that bound and use \eqref{eq:tau_to_zero_hard_assignment_limit} to conclude that
\begin{align}
    d_{\mathrm{perm}}^2\!\bigl(\bm{\mu}_{\mathrm{HA}}^\star(\sigma),\bm{\mu}^\star\bigr) \le 4K^2(K-1)\bigl(\Delta_{\max}^2+\sigma^2 d\bigr) \exp\!\left( -\Bigl(\frac{1}{8}-\eta\Bigr)\frac{\Delta_{\min}^2}{\sigma^2} \right)
\end{align}
for all sufficiently small $\sigma$. The convergence $d_{\mathrm{perm}}(\bm{\mu}_{\mathrm{HA}}^\star(\sigma),\bm{\mu}^\star)\to 0$
follows immediately.

\section{Latent-label clustering: proofs}

\subsection{Proof of Proposition~\ref{prop:low_snr_lb_classification}}
\label{sec:proof-lowSNR_classification}
Under Model~\ref{model:gmm_isotropic_equal}, let $P$ denote the joint law of $(L,Y)$, and let $P_L$ and $P_Y$ denote its marginals on $[K]$ and $\mathbb{R}^d$, respectively. Define
\begin{align}
    Q \triangleq P_L \otimes P_Y
\end{align}
to be the product measure on $[K]\times\mathbb{R}^d$, namely the unique probability measure satisfying
\begin{align}\label{eq:Q_rectangles}
    Q(B\times C)=P_L(B)\,P_Y(C)
\end{align}
for all measurable sets $B\subseteq [K]$ and $C\subseteq \mathbb{R}^d$.

For any measurable classifier $\psi:\mathbb{R}^d\to[K]$, define its success event
\begin{align}
    A_\psi \triangleq  \{(\ell,y)\in[K]\times\mathbb{R}^d:\ \psi(y)=\ell\}.
\end{align}
Its success probability under $P$ is $P_{\mathrm{succ}}(\psi)=P(A_\psi)$, and the Bayes success probability is
\begin{align}
    P_{\mathrm{succ}}^\star  \triangleq   \sup_{\psi} P_{\mathrm{succ}}(\psi),
    \qquad P_{\mathrm{err}}^\star  =  1-P_{\mathrm{succ}}^\star.
\end{align}
Since $Q=P_L\otimes P_Y$ and $P_L$ is uniform on $[K]$, we have, for every $\psi$,
\begin{align}
    Q(A_\psi)
    &= \sum_{\ell=0}^{K-1} \int_{\mathbb{R}^d} \mathbbm{1}\{\psi(y)=\ell\}\,Q(L=\ell)\,dP_Y(y) \\
    &= \sum_{\ell=0}^{K-1} \int_{\mathbb{R}^d} \mathbbm{1}\{\psi(y)=\ell\}\,\frac{1}{K}\,dP_Y(y)
     \\ & = \frac{1}{K}\int_{\mathbb{R}^d}\sum_{\ell=0}^{K-1}\mathbbm{1}\{\psi(y)=\ell\}\,dP_Y(y)
     = \frac{1}{K}.
\end{align}
Therefore, for any $\psi$, and using the definition of total-variation (TV) distance~\cite{cover1999elements}
\begin{align}
    P_{\mathrm{succ}}(\psi)-\frac{1}{K} = P(A_\psi)-Q(A_\psi)  \le \sup_{A} |P(A)-Q(A)| = \|P-Q\|_{\rm TV}.
\end{align}
Taking the supremum over $\psi$ yields
\begin{align}
    P_{\mathrm{succ}}^\star  \le  \frac{1}{K} + \|P-Q\|_{\rm TV},
\end{align}
and hence
\begin{align}\label{eq:TV_to_error}
    P_{\mathrm{err}}^\star = 1-P_{\mathrm{succ}}^\star \ge  1-\frac{1}{K}-\|P-Q\|_{\rm TV}.
\end{align}
By Pinsker's inequality~\cite[Lemma~2.5]{tsybakov2009introduction}, together with the standard identity
\begin{align}
    D_{\rm KL}(P\|Q) = D_{\rm KL}(P_{L,Y}\,\|\,P_L\otimes P_Y) = I(L;Y),
\end{align}
where $I(L;Y)$ is the mutual information between $L$ and $Y$ (with natural logarithms), we obtain from~\eqref{eq:TV_to_error} that
\begin{align}\label{eq:MI_to_error}
    P_{\mathrm{err}}^\star \ge 1-\frac{1}{K}-\sqrt{\frac{I(L;Y)}{2}}.
\end{align}
This proves~\eqref{eq:bayes_error_lb_in_terms_of_MI}.

\subsection{Proof of Corollary~\ref{cor:upper-bound-mutual-information}} \label{sec:corr-upper-bound-mutual-information}

Define the random vector $X \triangleq \mu_L$. Then $X$ is supported on $\{\mu_0,\dots,\mu_{K-1}\}$ and the observation model (Model~\ref{model:gmm_isotropic_equal}) can be written as the additive Gaussian-noise channel
\begin{align}\label{eq:AWGN_channel}
    Y = X + \sigma Z, \qquad Z\sim \mathcal{N}(0,I_d), 
\end{align}
with $Z\perp X$.
Since $X$ is a deterministic function of $L$, the data-processing identity~\cite{cover1999elements} implies
\begin{align}\label{eq:MI_L_equals_MI_X}
    I(L;Y) = I(X;Y).
\end{align}
Using the differential-entropy representation of mutual information,
\begin{align}\label{eq:MI_entropy_identity}
    I(X;Y) = h(Y)-h(Y|X) = h(X+\sigma Z)-h(\sigma Z),
\end{align}
where $h(\cdot)$ denotes differential entropy.

Let $\Sigma_X\triangleq \Cov(X)$. By~\eqref{eq:AWGN_channel},
\begin{align}\label{eq:CovY}
    \Cov(Y)=\Sigma_X+\sigma^2 I_d.
\end{align}
Among all random vectors with a fixed covariance, the Gaussian distribution maximizes differential entropy.
Therefore,
\begin{align}\label{eq:hY_gaussian_max}
    h(Y) \le \frac{1}{2}\log \Big((2\pi e)^d \det(\Sigma_X+\sigma^2 I_d)\Big).
\end{align}
Moreover, since $\sigma Z\sim \mathcal{N}(0,\sigma^2 I_d)$,
\begin{align}\label{eq:hSigmaZ}
    h(\sigma Z) = \frac{1}{2}\log \Big((2\pi e)^d \det(\sigma^2 I_d)\Big) = \frac{d}{2}\log(2\pi e\sigma^2).
\end{align}
Combining~\eqref{eq:MI_entropy_identity}--\eqref{eq:hSigmaZ} yields
\begin{align}\label{eq:logdet_bound}
    I(L;Y)=I(X;Y) & \le \frac{1}{2}\log\det(\Sigma_X+\sigma^2 I_d)-\frac{1}{2}\log\det(\sigma^2 I_d) \\ & = \frac{1}{2}\log\det \Big(I_d+\frac{1}{\sigma^2}\Sigma_X\Big).
\end{align}
Since $\Sigma_X\succeq 0$, all eigenvalues of $A\triangleq \sigma^{-2}\Sigma_X$ are nonnegative and
\begin{align}\label{eq:logdet_le_trace}
    \log\det(I_d+A) = \sum_{i=1}^d \log\big(1+\lambda_i(A)\big) \le \sum_{i=1}^d \lambda_i(A) = \tr(A),
\end{align}
where we used $\log(1+t)\le t$ for all $t\ge 0$. Applying~\eqref{eq:logdet_le_trace} to~\eqref{eq:logdet_bound},
\begin{align}\label{eq:MI_le_traceSigma}
    I(L;Y) \le \frac{1}{2}\tr \Big(\frac{1}{\sigma^2}\Sigma_X\Big) = \frac{1}{2\sigma^2}\tr(\Sigma_X).
\end{align}
Finally, since $L\sim\mathrm{Unif}([K])$ and $X=\mu_L$, we have
\begin{align}
    \mathbb{E}X = \frac{1}{K}\sum_{\ell=0}^{K-1}\mu_\ell = \bar{\mu},
\end{align}
where $\bar{\mu}$ is the mixture mean. Therefore
\begin{align}\label{eq:traceSigma_bound}
    \tr(\Sigma_X) &= \mathbb{E}\!\left[\|X-\mathbb{E}X\|^2\right] = \mathbb{E}\!\left[\|X-\bar{\mu}\|^2\right] \notag\\
    &= \frac{1}{K}\sum_{\ell=0}^{K-1}\|\mu_\ell-\bar{\mu}\|^2.
\end{align}
Combining~\eqref{eq:MI_le_traceSigma} and~\eqref{eq:traceSigma_bound} yields
\begin{align}\label{eq:MI_upper}
    I(L;Y) \le \frac{1}{2\sigma^2}\cdot \frac{1}{K}\sum_{\ell=0}^{K-1}\|\mu_\ell-\bar{\mu}\|^2,
\end{align}
which is inequality~\eqref{eq:MI_energy_upper_bound_raw}.

Substituting~\eqref{eq:MI_upper} into~\eqref{eq:MI_to_error} gives
\begin{align}\label{eq:prop32_final}
    P_{\mathrm{err}}^\star
    &\ge
    1-\frac{1}{K}-\sqrt{\frac{1}{2}I(L;Y)} \notag\\
    &\ge 1-\frac{1}{K} - \sqrt{\frac{1}{2}\cdot \frac{1}{2\sigma^2}\cdot \frac{1}{K}\sum_{\ell=0}^{K-1}\|\mu_\ell-\bar{\mu}\|^2
    } \notag\\
    &= 1-\frac{1}{K} - \frac{1}{2}\sqrt{\frac{1}{\sigma^2}\cdot \frac{1}{K}\sum_{\ell=0}^{K-1}\|\mu_\ell-\bar{\mu}\|^2},
\end{align}
which is exactly inequality~\eqref{eq:bayes_error_lb_energy}. This completes the proof.

\subsection{Proof of Proposition~\ref{prop:highSNR_exp_bound}}
\label{sec:proof-of-highSNR-classification}

Fix $\ell\in[K]$ and condition on the event $\{L=\ell\}$. Under Model~\ref{model:gmm_isotropic_equal}, we have $Y=\mu_\ell+\sigma Z$, where $Z\sim \mathcal{N}(0,I_d)$.
For any $j\neq \ell$, corresponding to a wrong class whose mean $\mu_j$ competes with the true mean $\mu_\ell$, define the pairwise confusion event
\begin{align}\label{eq:pairwise_conf_event}
    E_{\ell,j} \;\triangleq\; \Big\{\|Y-\mu_j\|^2 \le \|Y-\mu_\ell\|^2\Big\}.
\end{align}
Since the Bayes rule~\eqref{eq:bayes_is_nearest_mean} selects the closest mean, the overall error event is contained in the union of the pairwise confusion events:
\begin{align}\label{eq:error_union}
    \{\hat L(Y)\neq \ell\} \;\subseteq\; \bigcup_{j\neq \ell} E_{\ell,j}.
\end{align}
By the union bound and~\eqref{eq:error_union},
\begin{align}\label{eq:union_bound_cond}
    \mathbb{P}(\hat L(Y)\neq \ell \mid L=\ell) \;\le\; \sum_{j\neq \ell}\mathbb{P}(E_{\ell,j}\mid L=\ell).
\end{align}

We now compute $\mathbb{P}(E_{\ell,j}\mid L=\ell)$. Expanding the squared norms in~\eqref{eq:pairwise_conf_event},
\begin{align}
    \|Y-\mu_j\|^2 \le \|Y-\mu_\ell\|^2
    \quad \Longleftrightarrow \quad
    2\langle Y-\mu_\ell,\mu_j-\mu_\ell\rangle \ge \|\mu_j-\mu_\ell\|^2. \label{eq:hyperplane_form}
\end{align}
Since $Y-\mu_\ell=\sigma Z$, ~\eqref{eq:hyperplane_form} becomes
\begin{align}\label{eq:gaussian_projection_event}
    E_{\ell,j} = \Big\{\langle Z,\mu_j-\mu_\ell\rangle \ge \frac{\|\mu_j-\mu_\ell\|^2}{2\sigma}\Big\}.
\end{align}
Let $u_{\ell j}\triangleq (\mu_j-\mu_\ell)/\|\mu_j-\mu_\ell\|$. Since $Z\sim\cN(0,I_d)$,
\begin{align}\label{eq:proj_distribution}
    \langle Z,u_{\ell j}\rangle \sim \cN(0,1),
    \qquad
    \langle Z,\mu_j-\mu_\ell\rangle = \|\mu_j-\mu_\ell\|\,\langle Z,u_{\ell j}\rangle.
\end{align}
Therefore, writing $\Delta_{\ell j}=\|\mu_j-\mu_\ell\|$, we obtain
\begin{align}\label{eq:pairwise_prob_exact}
    \mathbb{P}(E_{\ell,j}\mid L=\ell) \; = \; \mathbb{P} \left(\langle Z,u_{\ell j}\rangle \ge \frac{\Delta_{\ell j}}{2\sigma}\right) \; \triangleq \;  \Phi_{\mathcal{N}} \left(-\frac{\Delta_{\ell j}}{2\sigma}\right).
\end{align}
where $\Phi_{\mathcal{N}}$ denotes the cumulative distribution function of the standard normal distribution $\mathcal{N}(0,1)$. Substituting~\eqref{eq:pairwise_prob_exact} into~\eqref{eq:union_bound_cond} yields
\begin{align}\label{eq:cond_err_bound_phi}
    \mathbb{P}(\hat L(Y)\neq \ell \mid L=\ell) \; \le \; \sum_{j\neq \ell}\Phi_{\mathcal{N}} \left(-\frac{\Delta_{\ell j}}{2\sigma}\right).
\end{align}
Averaging over $\ell$ and using $\mathbb{P}(L=\ell)=1/K$ gives the first inequality in~\eqref{eq:highSNR_pairwise_sum_bound}:
\begin{align}
    P_{\mathrm{err}}^\star = \mathbb{P}(\hat L(Y)\neq L)
    & = \frac{1}{K}\sum_{\ell=0}^{K-1}\mathbb{P}(\hat L(Y)\neq \ell \mid L=\ell)
    \\ & \le \frac{1}{K}\sum_{\ell=0}^{K-1}\sum_{j\neq \ell}\Phi_{\mathcal{N}} \left(-\frac{\Delta_{\ell j}}{2\sigma}\right).
    \label{eq:avg_bound_phi}
\end{align}

For the exponential bound, we use the standard Gaussian tail inequality
\begin{align}
    \label{eq:gaussian_tail}
    \Phi_{\mathcal{N}}(-t)\le \frac{1}{2}e^{-t^2/2},
    \qquad t\ge 0,
\end{align}
to obtain the second inequality in~\eqref{eq:highSNR_pairwise_sum_bound}. Finally, since
$\Delta_{\ell j}\ge \Delta_{\min}$ for all $\ell\neq j$, we have
\begin{align}
    \frac{1}{2K}\sum_{\ell=0}^{K-1}\sum_{j\neq \ell} \exp \left(-\frac{\Delta_{\ell j}^2}{8\sigma^2}\right) \le \frac{1}{2K}\sum_{\ell=0}^{K-1}\sum_{j\neq \ell} \exp \left(-\frac{\Delta_{\min}^2}{8\sigma^2}\right) = \frac{K-1}{2}\exp \left(-\frac{\Delta_{\min}^2}{8\sigma^2}\right),
\end{align}
which proves~\eqref{eq:highSNR_minsep_bound}.

\begin{remark}[Sharper prefactor via Mills' ratio]
\label{rem:mills_ratio_high_snr}
The exponential bound derived above is sufficient for our purposes, since it captures the correct decay rate uniformly over all pairs of centers. If one is interested in the more precise high-SNR asymptotic regime $t\to\infty$, then the Gaussian tail estimate in~\eqref{eq:gaussian_tail} can be sharpened using Mills' ratio:
\begin{align}
    \Phi_{\mathcal{N}}(-t) \sim \frac{1}{t\sqrt{2\pi}}e^{-t^{2}/2},
    \qquad t\to\infty.
\end{align}
Applying this with $t=\Delta_{\ell j}/(2\sigma)$ shows that each pairwise term satisfies
\begin{align}
    \Phi_{\mathcal{N}}\!\left(-\frac{\Delta_{\ell j}}{2\sigma}\right) \sim \frac{2\sigma}{\Delta_{\ell j}\sqrt{2\pi}} \exp\!\left(-\frac{\Delta_{\ell j}^{2}}{8\sigma^{2}}\right),
    \qquad \sigma\to 0.
\end{align}
Thus, the exponent in~\eqref{eq:highSNR_minsep_bound} is unchanged, but the prefactor can be refined from a crude constant to one of order $\sigma/\Delta_{\ell j}$. In particular, since $\Delta_{\ell j}\ge \Delta_{\min}$ for all $\ell\neq j$, one may replace the pairwise prefactor by the uniform upper bound $2\sigma/(\Delta_{\min}\sqrt{2\pi})$.
\end{remark}

\section{Symmetric two-component Gaussian mixture model: proofs}

Throughout this appendix we will repeatedly
invoke the following standard expansions of the error and complementary
error functions, available
in~\cite[Eqs.~7.1.2, 7.1.5, and 7.1.23]{abramowitz1964handbook}. As $x\to 0$,
\begin{align}
    \label{eq:erf_smallx}
    \mathrm{erf}(x)  &= \frac{2}{\sqrt{\pi}} \left(x - \frac{x^3}{3} + O(x^5)\right), \\
    \label{eq:erfc_smallx}
    \mathrm{erfc}(x) &= 1 - \frac{2}{\sqrt{\pi}} \left(x - \frac{x^3}{3} + O(x^5)\right),
\end{align}
and as $x\to\infty$,
\begin{align}
    \label{eq:erfc_largex}
    \mathrm{erfc}(x) = \frac{e^{-x^2}}{\sqrt{\pi}\,x} \left(1 - \frac{1}{2x^2} + O(x^{-4})\right).
\end{align}

\subsection{Reduction to a one-dimensional representation}

We begin by reducing the symmetric two-component model introduced in~\eqref{eq:K2_d_gmm_2} to a one-dimensional problem along the signal direction.

\begin{lem}[Reduction to a one-dimensional representation]
\label{lem:K2_reduction_to_1D}
Consider the centered two-component model $Y$ drawn according to~\eqref{eq:K2_d_gmm_2} with $\mu\neq 0$. Let $u\triangleq \mu/\|\mu\|_2\in\mathbb{S}^{d-1}$ and $T\triangleq \langle u,Y\rangle$. Then the population hard-assignment minimizer has the symmetric form $\bm{\mu}_{\mathrm{HA}}^\star = \bigl(\mu_{\mathrm{HA},1}^\star,\mu_{\mathrm{HA},2}^\star\bigr)$, with $\mu_{\mathrm{HA},2}^\star = -\mu_{\mathrm{HA},1}^\star$, and
\begin{align}
    \mu_{\mathrm{HA},1}^\star = u\,\mathbb{E}[T\mid T\ge 0].
\end{align}
\end{lem}

\begin{proof}[Proof of Lemma~\ref{lem:K2_reduction_to_1D}]
For antipodal mean $\pm au$ with $a\ge 0$, the Voronoi boundary is the hyperplane $\langle u,y\rangle=0$, since
\begin{align}
    \|y-au\|^2-\|y+au\|^2=-4a\langle u,y\rangle.
\end{align}
Hence the corresponding hard-assignment rule assigns $y$ to the first cluster if and only if $\langle u,y\rangle\ge 0$. The associated population mean is therefore
\begin{align}
    \mu_{\mathrm{HA},1}^\star = \mathbb{E}[Y\mid T\ge 0].
\end{align}
Now write $Y=Tu+W$, where $W\triangleq (I_d-uu^\top)Y\in u^\perp$. Also, writing $Y=S\mu+\sigma Z$, where $S\sim\mathrm{Unif}\{\pm1\}$ and $Z\sim\mathcal{N}(0,I_d)$, we have
\begin{align}
    T & = S\|\mu\|_2+\sigma\langle u,Z\rangle,
    \\
    W & = \sigma(I_d-uu^\top)Z.
\end{align}
By isotropy of $Z$, the scalar $\langle u,Z\rangle$ is independent of $(I_d-uu^\top)Z$. Hence $W$ is independent of $T$, and since $\mathbb{E}[W]=0$, it follows that $\mathbb{E}[W\mid T\ge 0]=0$. Therefore
\begin{align}
    \mu_{\mathrm{HA},1}^\star = \mathbb{E}[Y\mid T\ge 0] = u\,\mathbb{E}[T\mid T\ge 0].
\end{align}
By symmetry, $\mu_{\mathrm{HA},2}^\star=-\mu_{\mathrm{HA},1}^\star$.
\end{proof}

Lemma~\ref{lem:K2_reduction_to_1D} shows that $\mu_{\mathrm{HA},1}^\star = u\,\mathbb{E}[T\mid T\ge 0]$, where $T\sim \frac{1}{2}\,\mathcal{N}(\|\mu\|_2,\sigma^2) + \frac{1}{2}\,\mathcal{N}(-\|\mu\|_2,\sigma^2)$.
Since the density of $T$ is symmetric about $0$, we have $\mathbb{P}(T>0)=\frac{1}{2}$, and therefore
\begin{align}
    \mathbb{E}[T\mid T>0]=2\,\mathbb{E}[T\,\mathbf 1_{\{T>0\}}].
\end{align}
Moreover, by symmetry,
\begin{align}
    \mathbb{E}[T\mid T>0]=\mathbb{E}[|T|].
\end{align}
Thus $|T|$ follows the folded normal distribution with parameters $\|\mu\|_2$ and $\sigma^2$. By the folded-normal mean formula~\cite[Eq.~(10)]{leone1961folded} and the relation between $\mathrm{erf}$ and the standard normal CDF,
\begin{align}
    \label{eq:E:num_final}
    \mathbb{E}[|T|] = \|\mu\|_2\,\mathrm{erf}\!\left(\frac{\|\mu\|_2}{\sqrt{2}\sigma}\right) + \sigma\sqrt{\frac{2}{\pi}}\, \exp\!\left(-\frac{\|\mu\|_2^2}{2\sigma^2}\right).
\end{align}
Substituting this into the representation of $\mu_{\mathrm{HA},1}^\star$ from Lemma~\ref{lem:K2_reduction_to_1D}, and using $u=\mu/\|\mu\|_2$ together with $\mathrm{erf}=1-\mathrm{erfc}$, we obtain
\begin{align}
    \label{eq:c_k2_final}
    \mu_{\mathrm{HA},1}^\star = \mu -  \mu\,\mathrm{erfc}\!\left(\frac{\|\mu\|_2}{\sqrt{2}\sigma}\right) + \sqrt{\frac{2}{\pi}}\, \frac{\mu\sigma}{\|\mu\|_2}\, \exp\!\left(-\frac{\|\mu\|_2^2}{2\sigma^2}\right).
\end{align}
By symmetry,
\begin{align}
    \mu_{\mathrm{HA},2}^\star=-\mu_{\mathrm{HA},1}^\star.
\end{align}

\subsection{Proof of Theorem~\ref{thm:perm_MSE_d_K2}}
\label{app:MSEForK2}

For the symmetric two-component mixture~\eqref{eq:K2_d_gmm_2}, the population hard-assignment minimizer has the form $\bm{\mu}_{\mathrm{HA}}^\star=(\mu_{\mathrm{HA},1}^\star,\mu_{\mathrm{HA},2}^\star)$ with $\mu_{\mathrm{HA},2}^\star=-\mu_{\mathrm{HA},1}^\star$. Hence the normalized MSE from Definition~\ref{eq:mseDef_gmm} reduces to
\begin{align}
    \mathrm{MSE}(\bm{\mu}_{\mathrm{HA}}^\star,\bm{\mu}^\star) = \frac{1}{2\|\mu\|_2^2}\, d_{\mathrm{perm}}^2(\bm{\mu}_{\mathrm{HA}}^\star,\bm{\mu}^\star) = \frac{1}{\|\mu\|_2^2}\, \|\mu_{\mathrm{HA},1}^\star-\mu\|_2^2.
\end{align}
Substituting the explicit expression for $\mu_{\mathrm{HA},1}^\star$ from~\eqref{eq:c_k2_final} gives
\begin{align}
    \label{eq:1DMse}
    \mathrm{MSE}(\bm{\mu}_{\mathrm{HA}}^\star,\bm{\mu}^\star) = \frac{1}{\|\mu\|_2^2} \left(\sqrt{\frac{2}{\pi}}\,\sigma\, \exp\!\left(-\frac{\|\mu\|_2^2}{2\sigma^2}\right) - \|\mu\|_2\, \mathrm{erfc}\!\left(\frac{\|\mu\|_2}{\sqrt{2}\sigma}\right)\right)^2.
\end{align}

\subsection{Proof of Corollary~\ref{cor:K2_HA_asymptotics}}
\label{app:AsymptoticSigmaInfK2}

We prove the two asymptotic regimes by expanding the exact expression~\eqref{eq:1DMse}.

\paragraph{Step 1: $\sigma\to\infty$.}
As $\sigma\to\infty$, applying the standard Taylor expansions of $\exp(x)$ and $\mathrm{erfc}(x)$ as seen in~\eqref{eq:erfc_smallx} to~\eqref{eq:1DMse} yields
\begin{align}
    \sqrt{\frac{2}{\pi}}\,\sigma\, \exp\!\left(-\frac{\|\mu\|_2^2}{2\sigma^2}\right) - \|\mu\|_2\, \mathrm{erfc}\!\left(\frac{\|\mu\|_2}{\sqrt{2}\sigma}\right) = \sqrt{\frac{2}{\pi}}\,\sigma-\|\mu\|_2+O\!\left(\frac{1}{\sigma}\right).
\end{align}
Therefore
\begin{align}
    \mathrm{MSE}(\bm{\mu}_{\mathrm{HA}}^\star,\bm{\mu}^\star) = \frac{1}{\|\mu\|_2^2} \left(\sqrt{\frac{2}{\pi}}\,\sigma-\|\mu\|_2+O\!\left(\frac{1}{\sigma}\right)\right)^2,
\end{align}
and hence
\begin{align}
    \lim_{\sigma\to\infty} \frac{\mathrm{MSE}(\bm{\mu}_{\mathrm{HA}}^\star,\bm{\mu}^\star)}{\sigma^2} = \frac{2}{\pi\|\mu\|_2^2}.
\end{align}

\paragraph{Step 2: $\sigma\to0$.}
Set $\alpha\triangleq \|\mu\|_2/(\sqrt{2}\sigma)$, so that $\alpha\to\infty$ as $\sigma\to0$. Applying the asymptotic expansion of $\mathrm{erfc}(\alpha)$ as seen in \eqref{eq:erfc_largex}, with $\alpha=\|\mu\|_2/(\sqrt{2}\sigma)$, or equivalently $1/\alpha=\sqrt{2}\sigma/\|\mu\|_2$, gives
\begin{align}
    \|\mu\|_2\,\mathrm{erfc}(\alpha) = \sqrt{\frac{2}{\pi}}\,\sigma\,e^{-\alpha^2} \left( 1-\frac{\sigma^2}{\|\mu\|_2^2} +O\!\left(\frac{\sigma^4}{\|\mu\|_2^4}\right)\right).
\end{align}
Writing $A\triangleq \sqrt{2/\pi}\,\sigma\,e^{-\alpha^2}$, the bracket in~\eqref{eq:1DMse} becomes
\begin{align}
    A-\|\mu\|_2\,\mathrm{erfc}(\alpha) =A\left(\frac{\sigma^2}{\|\mu\|_2^2} +O\!\left(\frac{\sigma^4}{\|\mu\|_2^4}\right) \right) = \sqrt{\frac{2}{\pi}}\, \frac{\sigma^3}{\|\mu\|_2^2}\, e^{-\alpha^2} \left( 1+O\!\left(\frac{\sigma^2}{\|\mu\|_2^2}\right) \right).
\end{align}
Squaring and substituting into~\eqref{eq:1DMse} yields
\begin{align}
    \mathrm{MSE}(\bm{\mu}_{\mathrm{HA}}^\star,\bm{\mu}^\star) = \frac{2}{\pi}\, \frac{\sigma^6}{\|\mu\|_2^6} \left(1+O\!\left(\frac{\sigma^2}{\|\mu\|_2^2}\right)\right) e^{-\|\mu\|_2^2/\sigma^2},
    \qquad \sigma\to0.
\end{align}
Equivalently, since
$d_{\mathrm{perm}}^2(\bm{\mu}_{\mathrm{HA}}^\star,\bm{\mu}^\star)
= \|\bm{\mu}^\star\|_F^2\,\mathrm{MSE}
=2\|\mu\|_2^2\,\mathrm{MSE}$,
we obtain
\begin{align}
    d_{\mathrm{perm}}^2\!\bigl(\bm{\mu}_{\mathrm{HA}}^\star(\sigma),\bm{\mu}^\star\bigr) = \frac{4}{\pi}\, \frac{\sigma^6}{\|\mu\|_2^4} \left(1+O\!\left(\frac{\sigma^2}{\|\mu\|_2^2}\right) \right) e^{-\|\mu\|_2^2/\sigma^2},
    \qquad \sigma\to0.
\end{align}

\subsection{Proof of Proposition~\ref{prop:bayes_error_K2}}
\label{sec:proof_bayes_error_K2}

In the symmetric two-component model, the Bayes-optimal classifier assigns $Y$ to the nearer of $\pm\mu$, equivalently according to the sign of $T=\langle u,Y\rangle$, where $u=\mu/\|\mu\|_2$. Conditional on $L=1$, we have
$T\sim\mathcal{N}(\|\mu\|_2,\sigma^2)$, and therefore the corresponding misclassification probability is
\begin{align}
    \mathbb{P}(T<0\mid L=1) = \frac{1}{2}\,\mathrm{erfc}\!\left(\frac{\|\mu\|_2}{\sqrt{2}\sigma}\right).
\end{align}
By symmetry, the same expression holds for $\mathbb{P}(T>0\mid L=2)$. Since the two components have equal weight,
\begin{align}
    \label{eq:PerrK2Full}    P_{\mathrm{err}}^\star(\bm{\mu}^\star,\sigma) = \frac{1}{2}\, \mathrm{erfc}\!\left(\frac{\|\mu\|_2}{\sqrt{2}\sigma}\right).
\end{align}

\subsection{Proof of Corollary~\ref{cor:bayes_error_K2_asymptotics}}
\label{sec:proof_bayes_error_K2_asymptotics}

We derive the two asymptotic regimes from the exact expression~\eqref{eq:PerrK2Full}.

\paragraph{Step 1: Low SNR.}
Using the standard Taylor expansion of $\mathrm{erf}(x)$ as seen in~\eqref{eq:erf_smallx} with $x=\|\mu\|_2/(\sqrt{2}\sigma)$ into~\eqref{eq:PerrK2Full}, and using~\eqref{eq:SNR_K2_identity}, we obtain
\begin{align}
    P_{\mathrm{err}}^\star(\bm{\mu}^\star,\sigma) = \frac{1}{2} - \frac{1}{\sqrt{2\pi}}\sqrt{\mathrm{SNR}} + O\!\left(\mathrm{SNR}^{3/2}\right),
    \qquad \mathrm{SNR}\to0.
\end{align}

\paragraph{Step 2: High SNR.}
Using the asymptotic expansion of $\mathrm{erfc}(x)$ as seen in \eqref{eq:erfc_largex}, and substituting $x=\|\mu\|_2/(\sqrt{2}\sigma)=\sqrt{\mathrm{SNR}/2}$ into~\eqref{eq:PerrK2Full}, we obtain
\begin{align}
    P_{\mathrm{err}}^\star(\bm{\mu}^\star,\sigma) = \frac{1}{\sqrt{2\pi\,\mathrm{SNR}}}\exp\!\left(-\frac{\mathrm{SNR}}{2}\right) (1+o(1)),
    \qquad \mathrm{SNR}\to\infty.
\end{align}

\end{appendices}
\end{document}